\documentclass{article}
\usepackage[margin=1.2in]{geometry}

\usepackage{graphicx}
\usepackage{enumerate}
\usepackage{verbatim}
\usepackage{amsmath}
\usepackage{amsfonts}
\usepackage{amsthm}
\usepackage{url}
\usepackage{xspace}
\usepackage{color}
\usepackage[labelfont=bf,size=small]{caption}
\usepackage[labelfont=normalfont,justification=centering]{subcaption}
\usepackage{algo}
\usepackage{hyperref}
\usepackage{authblk}

\newcounter{dummycount}

\newtheorem{theorem}{Theorem}
\newtheorem{lemma}[theorem]{Lemma}
\newtheorem{corollary}[theorem]{Corollary}

\newtheorem{claim}{Claim}

\newenvironment{claimproof}[1]{%
\noindent%
\textit{Proof of Claim #1.}%
}
{
\hfill $\triangle$%
\medskip
}

\newcommand{\wormhole}[1]
{
\newcounter{#1}
\setcounter{#1}{\value{theorem}}
}

\newenvironment{backInTime}[1]
{
\setcounter{dummycount}{\value{theorem}}
\setcounter{theorem}{\value{#1}}
}
{
\setcounter{theorem}{\value{dummycount}}
}

\title{Contact Representations of Graphs in 3D}
\author[1]{Md. Jawaherul Alam}
\author[2]{William Evans}
\author[1]{Stephen G. Kobourov}
\author[1]{Sergey Pupyrev}
\author[1]{Jackson Toeniskoetter}
\author[3]{Torsten Ueckerdt}
\affil[1]{Department of Computer Science, University of Arizona, USA%\\
}
\affil[2]{Department of Computer Science, University of British Columbia, Canada%\\
}
\affil[3]{Department of Mathematics, Karlsruhe Institute of Technology, Germany%\\
}
\newcommand{\TT}{{\mathcal{T}}}

\newcommand{\R}{\mathbb R}

\newcommand{\df}{\textit}
\newcommand{\WLOG}{w.l.o.g.\xspace}
\newcommand{\LL}{$\mathcal{L}$\xspace}
\newcommand{\LLs}{$\mathcal{L}$'s\xspace}

\newcommand{\pperp}{\ensuremath{(\bot{-}\bot)}}
\newcommand{\ppar}{\ensuremath{(\bot{-}||)}}

\newcommand{\Oh}{{\ensuremath{\mathcal{O}}}}

\begin{document}

\date{}

\maketitle

%\begin{bibunit}

\begin{abstract}
We study contact representations of graphs in which vertices are represented by axis-aligned polyhedra in 3D and edges are realized by non-zero area common boundaries between corresponding polyhedra. We show that for every 3-connected planar graph, there exists a simultaneous representation of the graph and its dual with 3D boxes. We give a linear-time algorithm for constructing such a representation. This result extends the existing primal-dual contact representations of planar graphs in 2D using circles and triangles.
While contact graphs in 2D directly correspond to planar graphs, we next study representations of non-planar graphs in 3D. In particular we consider representations of optimal 1-planar graphs. A graph is 1-planar if there exists a drawing in the plane where each edge is crossed at most once, and an optimal $n$-vertex 1-planar graph has the maximum ($4n-8$) number of edges.  We describe a linear-time algorithm for representing optimal 1-planar graphs without separating 4-cycles with 3D boxes. However, not every optimal 1-planar graph admits a representation with boxes. Hence, we consider contact representations with the next simplest axis-aligned 3D object, L-shaped polyhedra. We provide a quadratic-time algorithm for representing optimal 1-planar graph with L-shaped polyhedra.
\end{abstract}

\section{Introduction}

Graphs are often used to capture relationships between objects, and graph embedding
techniques allow us to visualize such relationships via traditional node-links diagrams.
There are compelling theoretical and practical reasons to study \df{contact representations}
of graphs, where vertices are geometric objects and edges correspond to pairs of objects
touching in some specified fashion. In practice, 2D contact
representations with rectangles, circles, and polygons of low
complexity are intuitive, as they provide the viewer with the familiar metaphor
of geographical maps. Such representations are preferred in some contexts over the standard
node-link representations for displaying relational information~\cite{BGPV08,GHK10}.
Contact representations of graphs have practical applications in
data visualization~\cite{Shn92a}, %,Bo10},
cartography~\cite{Rai34}, geography~\cite{Tob04a},
sociology~\cite{HK98},
very-large-scale integration circuit design~\cite{Ull84}, and floor-planning~\cite{MCP02a}.

A large body of work considers representing graphs as contacts of
simple curves or polygons in 2D.
Graphs that can be represented in this way are planar.
In fact, Koebe's 1936 theorem established that \emph{all} planar
graphs can be represented by touching disks~\cite{Koe36}.
Any planar graph also has a contact representation with
triangles~\cite{FMR94,GLP12}. Partial planar $3$-trees and some
series-parallel graphs admit a representation with homothetic
triangles~\cite{BBGDF+07}. Curve contacts~\cite{Hli98}, line-segment
contacts~\cite{FM07a,KM94}, and $L$-shape contacts~\cite{CKU13,KUV13}
have also been used.
In particular, it is known that all planar bipartite graphs can be represented by axis-aligned
segment contacts~\cite{CKU98,FMP91,RT86}.
%They can also be represented by contact of axis-aligned line segments, $L$-shapes, and $\Gamma$-shapes~\cite{CU13}.
Furthermore, any planar graph has a representation with $T$-shapes~\cite{FMR94}.
Despite our best efforts, however, many graphs remain stubbornly
non-planar and, for these, such contact representations in 2D are impossible.
Hence, a natural generalization is representing vertices with 3D-polyhedra, such
as cubes and tetrahedra, and edges with shared boundaries. While such contact
representations allow us to visualize non-planar graphs, there is much less known
about contact representations in 3D than in 2D. Contact graphs using 3D objects
have been studied for complete graphs and complete bipartite graphs
using spheres~\cite{BR13,HK01} and cylinders~\cite{Bez05}.
%, and tetrahedra~\cite{ADIO13}.
%It is also known that recognition of contact graphs of unit balls in 3D is NP-hard~\cite{HK01}.

As a first step towards representing non-planar graphs, we consider
\df{primal-dual} contact representations, in which a plane graph (a planar graph with a fixed planar embedding)
and its dual are represented simultaneously. More formally, in such a representation
vertices and faces are represented by some geometric objects so that:

\begin{enumerate}[(i)]
	\item the objects for the vertices are interior-disjoint and induce a contact representation
		for the primal graph;
	\item the objects for the faces are interior-disjoint except for the object for the outer
		face, which contains all the objects for the internal faces, and together they induce a
		contact representation of the dual graph;
	\item the objects for a vertex $v$ and a face $f$ of the primal graph intersect if and only
		if $v$ and $f$ are incident.
\end{enumerate}

Primal-dual representations of planar graphs have been studied in 2D. Specifically,
every $3$-connected plane graph has a primal-dual representation with circles~\cite{And70} and
triangles~\cite{GLP12}; see Fig.~\ref{fig:intro}(a)--(c). Our first result in this paper
is an analogous primal-dual representation using axis-aligned 3D boxes. While it is known that every planar graph has a contact representation with
% interior-disjoint
3D boxes~\cite{Tho88,FF11,BEFHK+12},
 Theorem~\ref{thm:box} strengthens these results; see Fig.~\ref{fig:intro4}.

\wormhole{thm-box}
\begin{theorem}
\label{thm:box}
Every $3$-connected planar graph $G=(V,E)$ admits a proper primal-dual box-contact
representation in 3D and it can be computed in $\Oh(|V|)$ time.
\end{theorem}

\begin{figure}[t]
%\vspace{-1cm}
  \begin{subfigure}[t]{.2\textwidth}
    \centering
    \includegraphics[height=2.5cm]{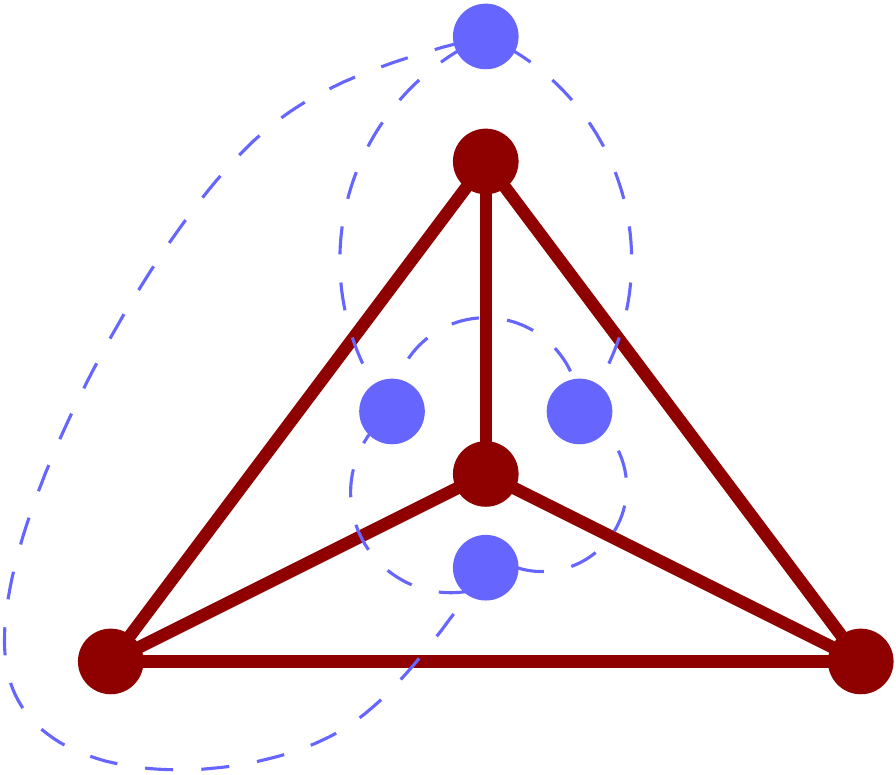}
    \caption{}
  \end{subfigure}
  \hfill
  \begin{subfigure}[t]{.22\textwidth}
    \centering
    \includegraphics[height=2.75cm]{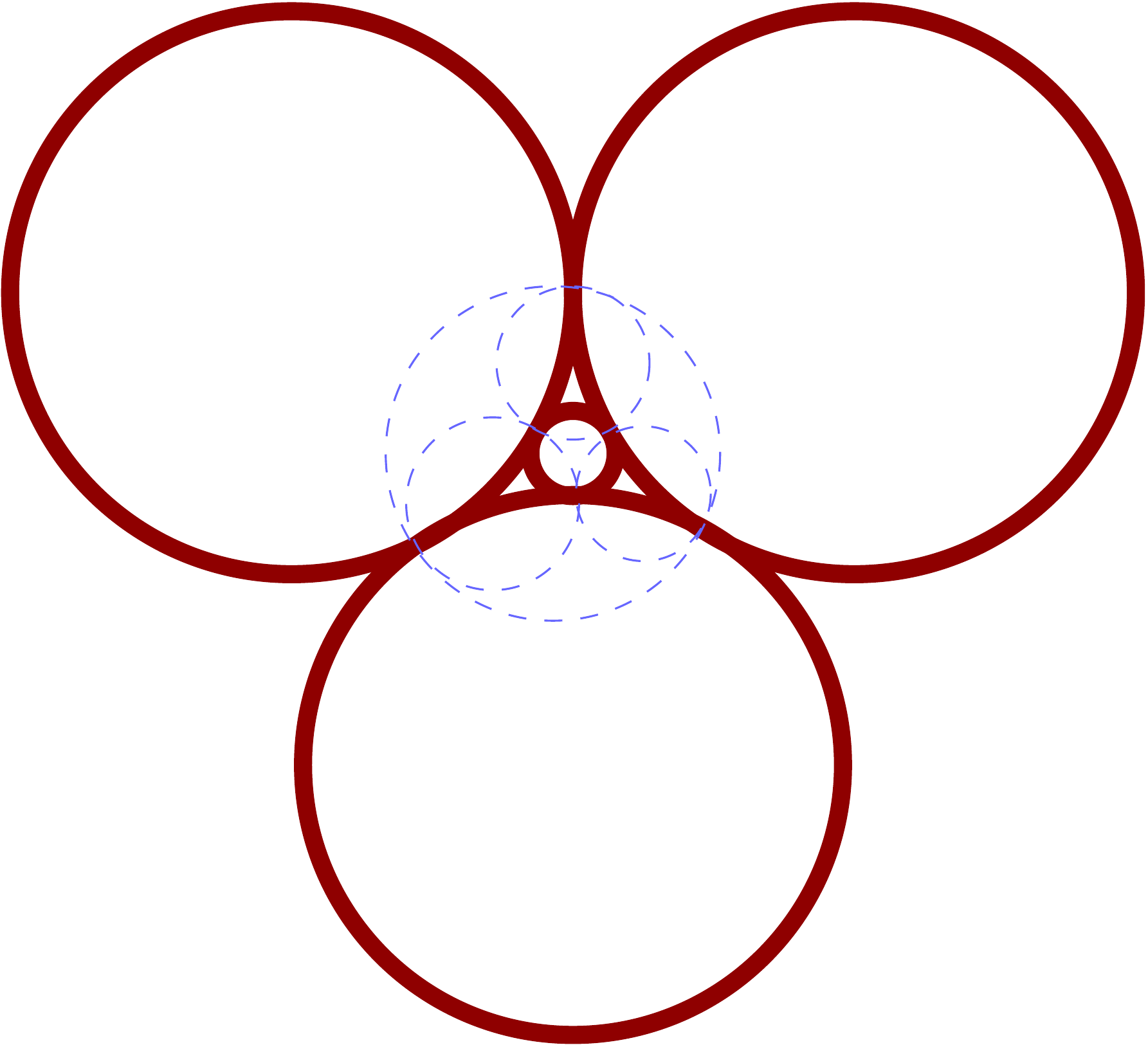}
    \caption{}
    \label{fig:intro2}
  \end{subfigure}
  \hfill
  \begin{subfigure}[t]{.22\textwidth}
    \centering
    \includegraphics[height=2.75cm]{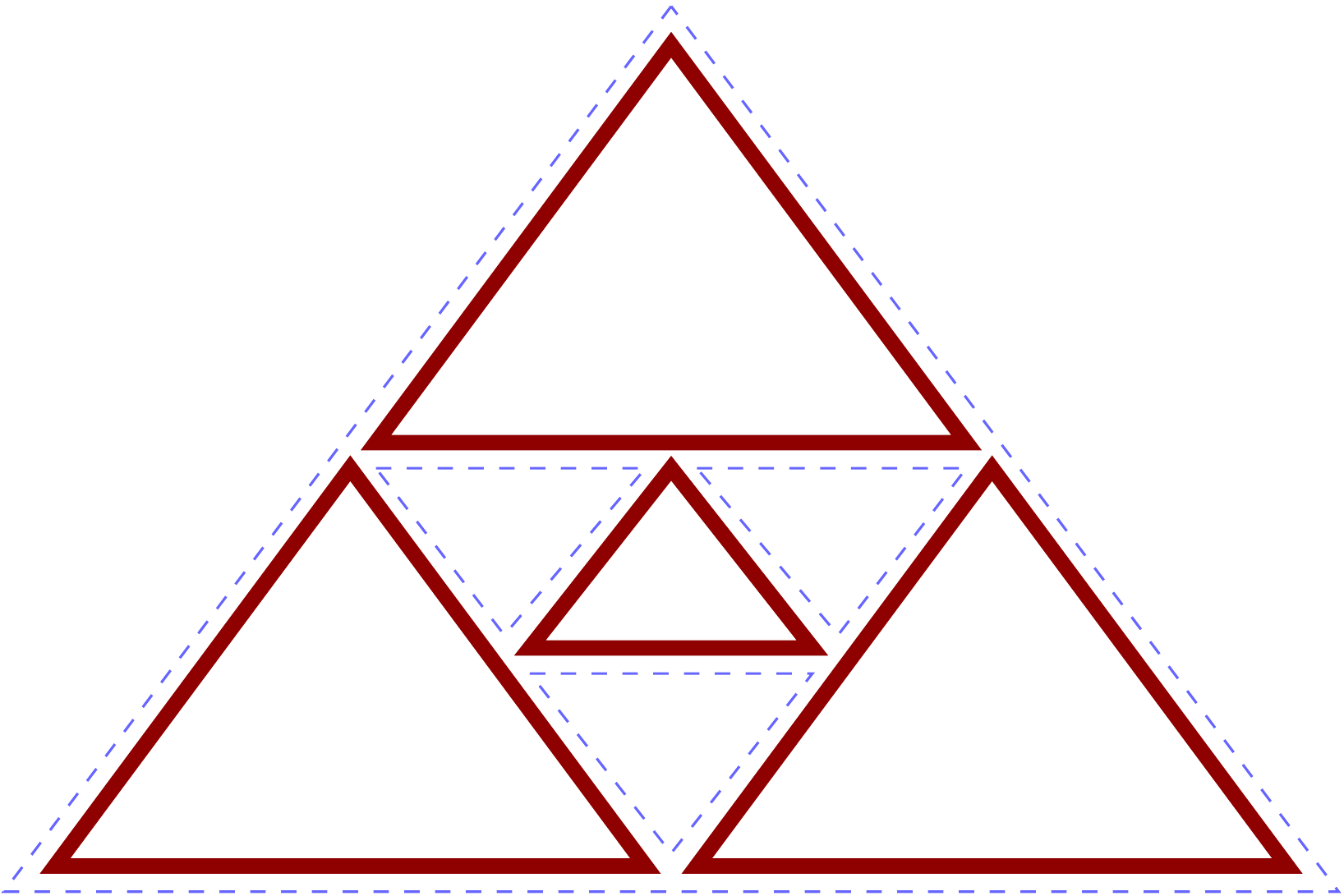}
    \caption{}
  \end{subfigure}
  \hfill
  \begin{subfigure}[t]{.3\textwidth}
    \centering
    \includegraphics[height=4cm]{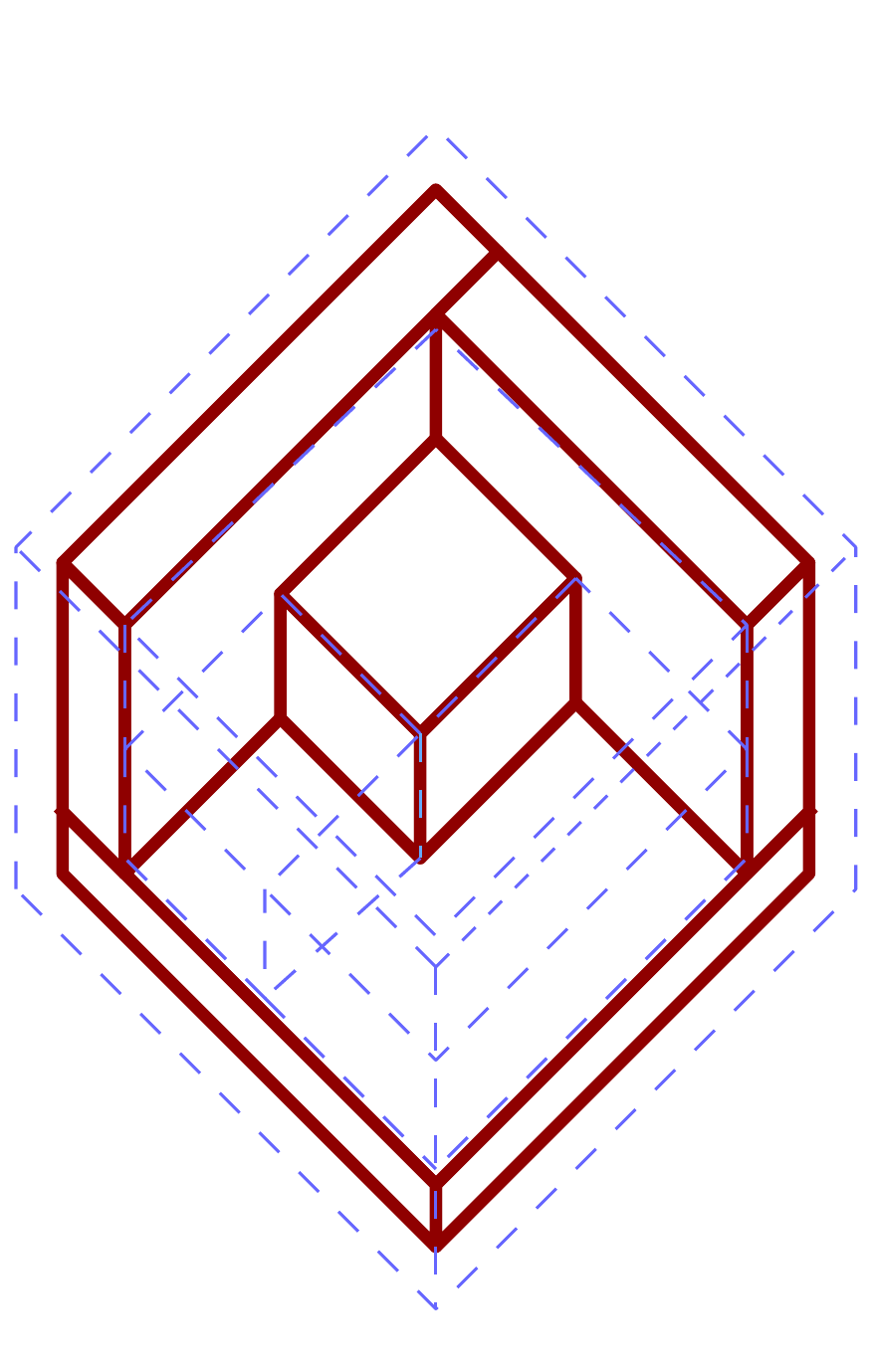}
    \caption{}
    \label{fig:intro4}
  \end{subfigure}
    \caption{(a)~A planar graph $K_4$ and its dual; primal-dual contact representations
     of the graph with (b)~circles and (c)~triangles. (d)~The primal-dual box-contact representation of $K_4$
     with dual vertices shown dashed. The outer box (shell) contains all other boxes.}
    \label{fig:intro}
\end{figure}

We would like to mention two differences with the existing primal-dual
 representations~\cite{And70,GLP12}. First, both of these existing constructions induce \df{non-proper} (point)
 contacts, while our contacts are always \df{proper}, that is, have non-zero areas.
Second, for a given $3$-connected plane graph, it is not always possible to find a primal-dual
 representation with circles by a polynomial-time algorithm~\cite{BDEG14}, although it can be
 constructed numerically by polynomial-time iterative schemes~\cite{CS03,Moh97}.
There is also no known polynomial-time algorithm that computes a primal-dual representation
 with triangles for a given $3$-connected plane graph.
%The representation with triangles~\cite{GLP12} is based on
Our representation, in contrast, can be found in linear time and can be
realized on the $n\times n\times n$ grid, where $n$ is the number of vertices of the graph.
%TODO: SAY that STRETCHABILITY is also not polynomial!!

We prove Theorem~\ref{thm:box} with two different constructive algorithms. The first one
uses the notions of Schnyder woods and orthogonal surfaces, as defined in~\cite{FZ08}.
Although it is known that every $3$-connected planar graph induces an orthogonal surface, we show how to
construct a new contact representation with interior-disjoint boxes from the orthogonal surface.
Since the orthogonal surfaces for a $3$-connected planar graph and its dual coincide
topologically,
the primal and the dual box-contact representations can be fit together to realize the desired contacts.
%(which is shown in \cite{FZ08}), allows us to fit together
The alternative algorithm
% (described in Appendix~\ref{app:alternateproof})
 builds a box-contact representation for a maximal planar graph using the notion of a canonical order
 of planar graph~\cite{FPP90}. Both construction ideas are inspired by recent box-contact representation
 algorithms for maximal planar graphs~\cite{BEFHK+12}; however, we generalize the algorithms to
 accommodate $3$-connected planar graphs and show that the primal and dual representations can
 be combined together.
Our methods rely on a correspondence between Schnyder woods and generalized canonical orders for
 $3$-connected planar graphs. Although the correspondence has been claimed earlier in~\cite{BBC11},
 the earlier proof appears to be incomplete. As another contribution of the paper, we provide a complete
 proof of the claim in Section~\ref{subsect:corr}. 
% see more detailed discussion in Section~\ref{subsect:corr}.

Theorem~\ref{thm:box} immediately gives box-contact representations for a special class of
non-planar graphs that are formed by the union of a planar graph and its dual. The graphs are
called \df{prime} and are defined as follows. A simple graph $G=(V,E)$ is said to be \df{$1$-planar}
if it can be drawn on the plane so that each of its edges crosses at most one other edge.
This class of graphs was first considered by Ringel~\cite{Rin65} in the context of simultaneously coloring a planar graph and its dual. A $1$-planar graph has at most $4|V| - 8$ edges~\cite{FM07,PT97}
and it is \df{optimal} if it has exactly $4|V| - 8$ edges, that is, it is the
densest $1$-planar graph on the vertex set~\cite{BSW84,BEGGH+13}. An optimal
$1$-planar graph is called \df{prime} if it has no
separating $4$-cycles, that is, cycles of length $4$ whose removal disconnects the graph. These optimal $1$-planar
graphs are exactly the ones that are $5$-connected and the ones that can be obtained as the union of a $3$-connected
simple plane graph, its dual and its vertex-face-incidence graph~\cite{Sch86}.
Note that in our primal-dual representation not all boxes are interior-disjoint, as one of the boxes
contains all other boxes. We call this special box the \df{shell} and such a
representation a \df{shelled} box-contact representation. Here all the vertices are represented
by 3D boxes, except for one vertex, which is represented by a shell, and the interiors of all the
boxes and the exterior of the shell are disjoint. Note that a similar shell
is also required in the circle-contact and triangle-contact
representations; see~Fig.~\ref{fig:intro}. The following is a direct
corollary of Theorem~\ref{thm:box}.

\wormhole{th-prime}
\begin{corollary}
\label{th:prime}
Every prime $1$-planar graph $G=(V,E)$ admits a shelled box-contact representation
in 3D and it can be computed in $\Oh(|V|)$ time.
\end{corollary}

\begin{figure}
%\begin{wrapfigure}[5]{r}{.2\textwidth}
%\vspace{-1cm}
  \centering
    \includegraphics[height=3cm]{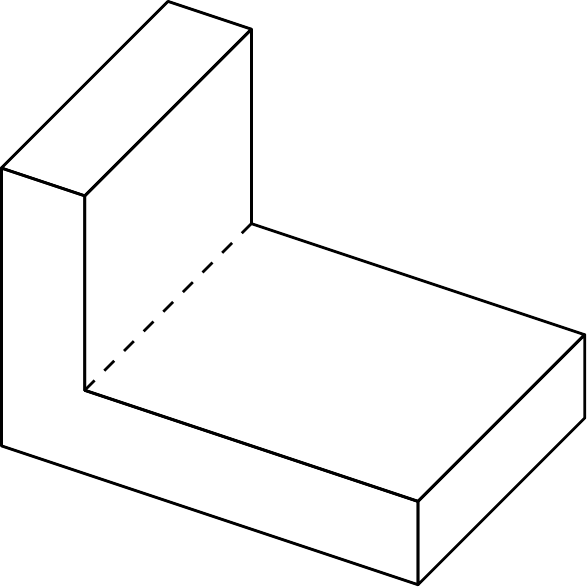}
  \caption{An $L$-shaped polyhedron.}
  \label{fig:L}
%\end{wrapfigure}
\end{figure}
One may wonder whether every $1$-planar graph admits a box-contact representation in 3D, but it is easy to see that there are $1$-planar graphs, even as simple as $K_5$, that do not admit a
box-contact representation. Furthermore, there exist optimal $1$-planar graphs (which contain separating $4$-cycles) that have neither a box-contact representation nor a shelled box-contact representation.
%; for details see Appendix~\ref{sec:proofs}.

Therefore, we consider representations with the next simplest
axis-aligned object in 3D, an $L$-shaped polyhedron
or simply an \LL, which is an axis-aligned box minus the intersection of
two axis-aligned half-spaces; see Fig.~\ref{fig:L}. Such an object can also be
considered as the union of two boxes in 3D. In the paper, we provide a quadratic-time algorithm
for representing every optimal 1-planar graph with~\LLs.

\wormhole{thm-1-planar}
\begin{theorem}
\label{thm:1-planar}
 Every optimal $1$-planar graph $G=(V,E)$ has a proper \LL-contact representation in 3D and it can be computed in $\Oh(|V|^2)$ time.
\end{theorem}

Our algorithm is similar to a recursive procedure used for constructing
box-contact representations of planar graphs in~\cite{FF11,Tho88}. The basic idea is to find
separating $4$-cycles and represent the inner and the outer parts of the graph induced by the cycles
separately. Then these parts are combined together to produce the final representation.
Since the separating $4$-cycles can be nested inside each other, the
running time of our algorithm is dominated by the finding of separating $4$-cycles
and the nested structure among them.
Unlike the early algorithms for box-contact representations of planar graphs~\cite{FF11,Tho88}, our algorithms produce proper contacts between the 3D objects (boxes and \LLs).
%are axis-aligned and all the contacts are proper.

%\section{Tools}
\section{Preliminaries}
\label{sec:prelim}

Here we introduce the tools needed to prove our results. In Section~\ref{subsect:def} we define
the known concepts of an \df{ordered path partition} and a \df{Schnyder wood}. In Section~\ref{subsect:corr} we describe new results about the relationship between these two structures for  $3$-connected plane graphs. Section~\ref{sec:ortho} reviews the concept of an \df{orthogonal surface}.

%\vspace{-0.1cm}
\subsection{Ordered Path Partitions, Canonical Orders and Schnyder Woods}
\label{subsect:def}

The concepts of a Schnyder wood and a canonical order
% (defined in Appendix~\ref{sec:proofs})
 were initially introduced for maximal plane graphs~\cite{FPP90,Sch90}; later they were generalized
 to $3$-connected plane graphs~\cite{Kan96,FZ08}.
Although the concepts are proved to be equivalent for maximal plane graphs~\cite{FM01},
 they are no longer equivalent after the generalization. Thus Badent~et~al.~\cite{BBC11}
 generalized the notion of a canonical order to an {\em ordered path partition} for a $3$-connected
 plane graph in an attempt to make it equivalent to a Schnyder wood.

Let $G$ be a $3$-connected plane graph with a specified pair of vertices $\{v_1, v_2\}$ and a third vertex
$v_3\notin\{v_1,v_2\}$, such that $v_1$, $v_2$, $v_3$ are all on the outer face in that counterclockwise
order. Add the edge $(v_1,v_2)$ to the outerface of $G$ (if it does not already contain it) and call the augmented graph
$G'$. Let $\Pi=(V_1,V_2,\ldots,V_L)$ be a partition of the vertices of $G$.
Then $\Pi$ is an \df{ordered path partition} of $G$ if the following conditions hold:

\begin{enumerate}[(1)]
%\vspace{-0.2cm}
%	\item $V_1$ consists of
	% $\{v_1,v_2\}$ and
%	the vertices on the clockwise path from $v_1$ to $v_2$ on the outer cycle of $G$;
%		$V_L=\{v_3\}$;

	\item $V_1$ contains the vertices on the clockwise path from $v_1$ to $v_2$ on the outer
		cycle;
		$V_L=\{v_3\}$;
	\item for $1\le k\le L$, the subgraph $G_k$ of $G'$ induced by the vertices in
		$V_1\cup\ldots\cup V_k$ is $2$-connected and internally $3$-connected
		(that is, removing two internal vertices of $G_k$ does not disconnect it); and
		the outer cycle $C_k$ of $G_k$ contains the edge $(v_1,v_2)$;

	\item for $2\le k\le L$, each vertex on $C_{k-1}$ has at most one neighbor on $V_k$.
\end{enumerate}

The pair of vertices $(v_1, v_2)$ forms the \df{base-pair} for $\Pi$ and $v_3$ is called the \df{head vertex}
of~$\Pi$.
For an ordered path partition $\Pi=(V_1,V_2,\ldots,V_L)$ of $G$, we say that a vertex $v$ of $G$ has
\df{level} $k$ if $v\in V_k$. The \df{predecessors} of $v$ are all the vertices with equal or smaller levels
and the \df{successors} of $v$ are all vertices with equal or larger levels; see Fig.~\ref{fig:can-schB}.
%Clearly, any canonical order in a $3$-connected plane graph $G$ is also an ordered path partition in $G$.
%These two definitions coincide for a maximal plane graph where for $2\le k\le L$, there is exactly one
%vertex with level $k$.

The definition of a \df{canonical order} is similar to that for an ordered path partition, but for Condition (3),
 which is replaced by the following more restricted Condition (3'):

\begin{enumerate}
	\item[(3')] for $2\le k\le L$, $V_k$ is a (left-to-right) path $\{v_a, v_{a+1}, \ldots v_b\}$, which is a
		subpath of $C_{k}-(v_1,v_2)$,
		and each $v_i$ with $i\in\{a,\dots, b\}$ has at least one neighbor in $G-G_k$.
		If $a<b$ then each of $v_a$ and $v_b$ has one neighbor on $C_{k-1}$ and these are the only
		neighbors of $V_k$ in $G_{k-1}$.
		If $a=b$ then $v_a=v_b$ may have multiple neighbors on $C_{k-1}$.
\end{enumerate}

Thus a canonical order is a special ordered path partition and the definition of the \df{base-pair},
 \df{head vertex}, \df{predecessors} and \df{successors} follow from that in an ordered path partition.

A Schnyder wood is defined as follows. Let $v_1$, $v_2$, $v_3$ be three specified vertices
in that counterclockwise order on the outer face of $G$. For $i=1,2,3$, add a half-edge from $v_i$
reaching into the outer face. Then a \df{Schnyder wood} is an orientation and coloring of all the edges of~$G$
(including the added half edges) with the colors $1,2,3$ satisfying the following conditions:

\begin{enumerate}[(1)]
%\vspace{-0.2cm}

	\item every edge $e$ is oriented in either one (\df{uni-directional edge}) or two opposite directions
		(\df{bi-directional edge}). The edges are colored and if $e$ is bi-directional, then the two
		directions have distinct colors;

	\item the half-edge at $v_i$ is directed outwards and colored $i$;

	\item each vertex $v$ has out-degree exactly one in each color, and the counterclockwise
		order of edges incident to $v$ is: outgoing in color $1$, incoming in color $2$, outgoing
		in color $3$, incoming in color $1$, outgoing in color $2$, incoming in color $3$;

	\item there is no interior face whose boundary is a directed cycle in one color.

\end{enumerate}

\begin{figure}[t]
%\vspace{-0.4cm}
\centering
  \begin{subfigure}[t]{.45\textwidth}
    \centering
    \includegraphics[height=6.8cm]{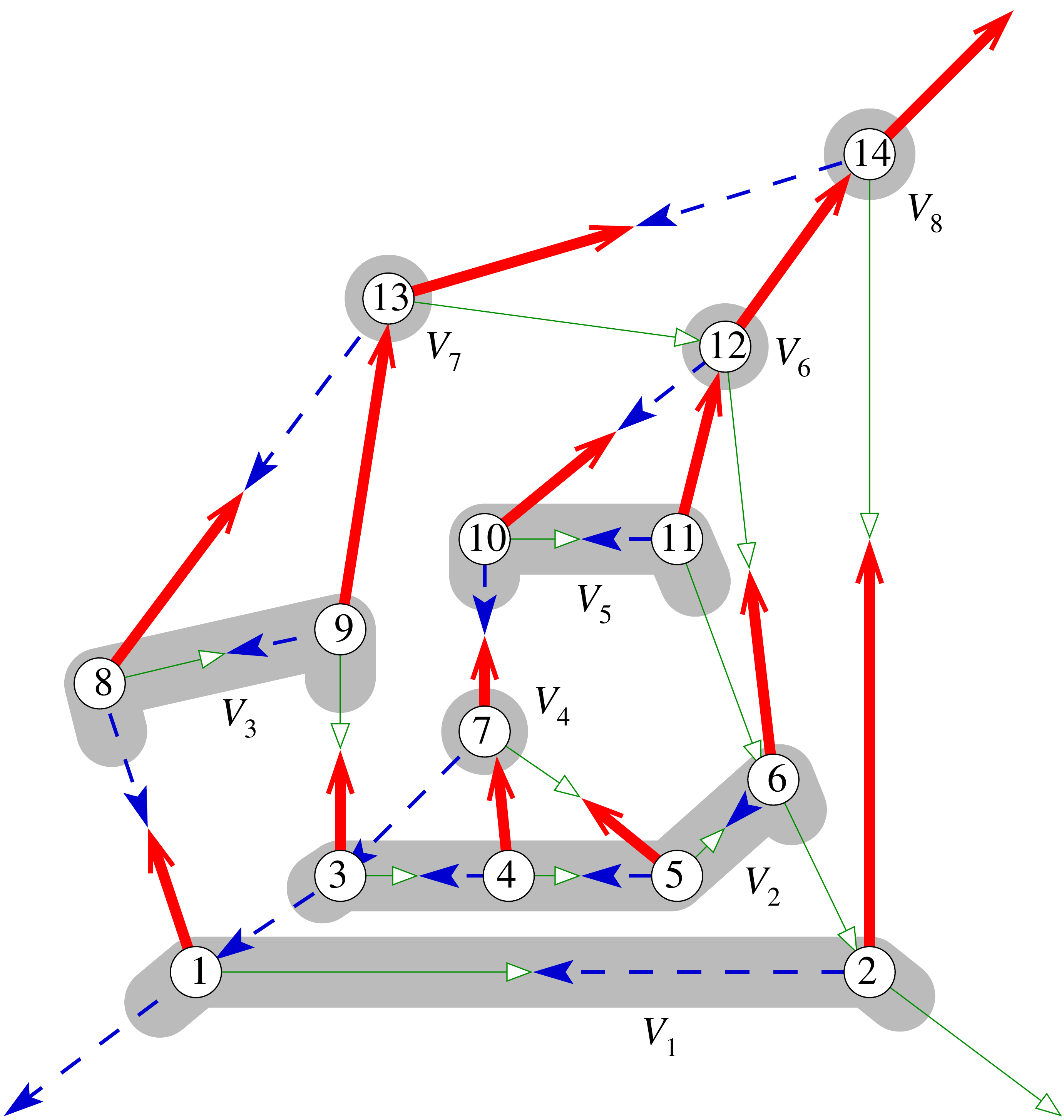}
    \caption{}
    \label{fig:can-schB}
  \end{subfigure}
~
  \begin{subfigure}[t]{.45\textwidth}
    \centering
    \includegraphics[height=6.8cm]{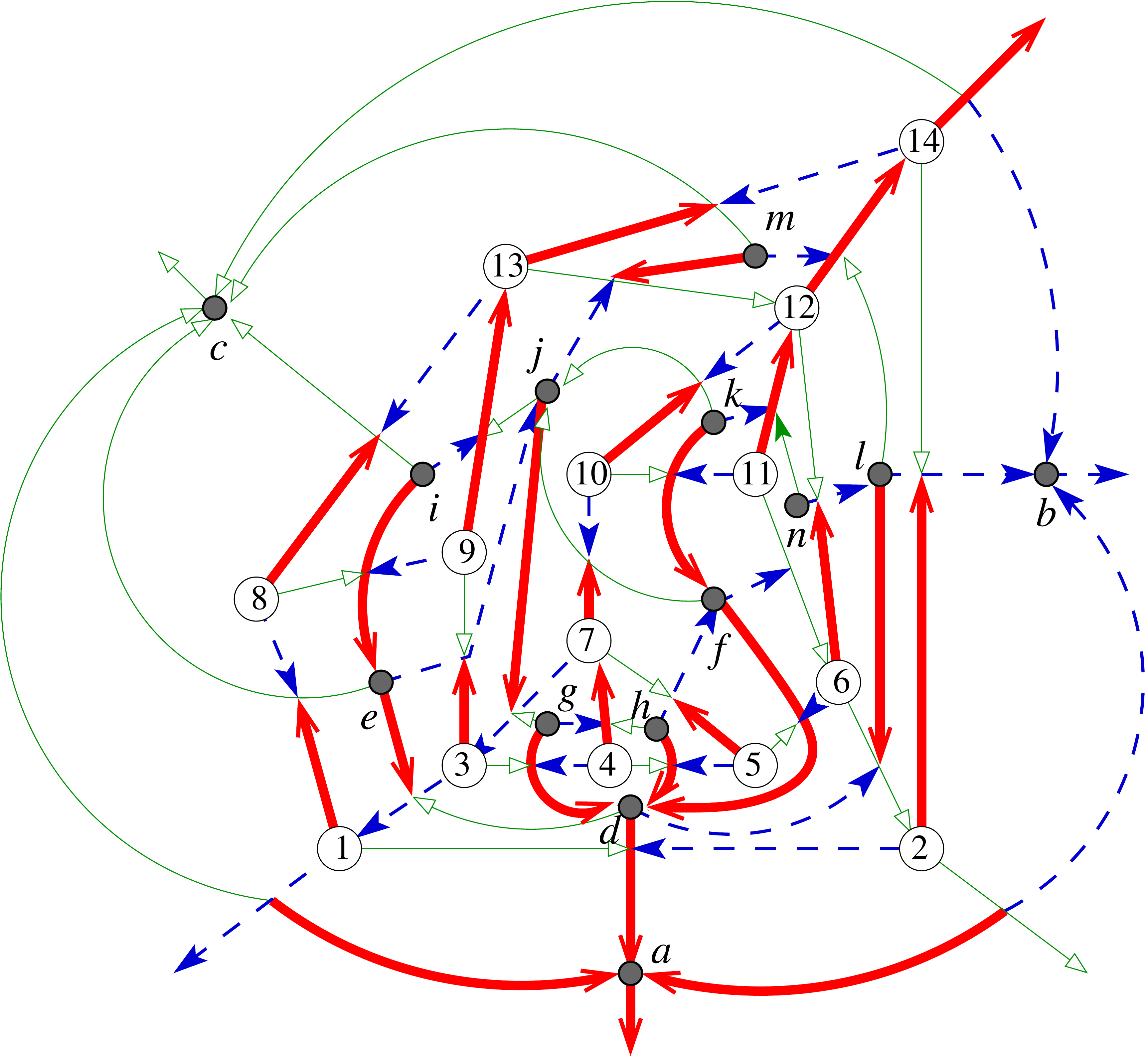}
	\caption{}
    \label{fig:can-schA}
  \end{subfigure}
\caption{(a)~An ordered path partition
% (also an elementary canonical order)
 and its corresponding Schnyder wood for a $3$-connected graph $G$.
 (b)~The Schnyder woods for the primal and the dual of $G$. The thick solid red, dotted blue and thin
 solid green edges represent the three trees in the Schnyder wood.}
\label{fig:can-sch}
\end{figure}
% Coordinates for boxes in the preceding figure:
% Blue/thick=1, Green/hollow=2, Red/dash=3
%  1: (11,0,0)(12,11,2)
%  2: (0,11,0)(11,12,11)
%  3: (8,2,1)(11,5,2)
%  4: (5,5,1)(8,6,3)
%  5: (4,6,1)(5,9,3)
%  6: (1,9,1)(4,11,6)
%  7: (4,4,3)(8,6,4)
%  8: (9,0,2)(11,1,9)
%  9: (8,1,2)(9,2,9)
% 10: (3,4,4)(4,5,6)
% 11: (2,5,4)(3,9,6)
% 12: (1,4,6)(3,9,11)
% 13: (2,0,9)(9,4,11)
% 14: (0,0,11)(2,11,12)
%  a: (,,)(11,11,0)
%  b: (,,)(0,11,11)
%  c: (,,)(11,0,11)
%  d: (2,2,0)(11,11,1)
%  e: (8,0,1)(11,2,2)
%  f: (2,5,1)(4,9,4)
%  g: (5,4,1)(8,5,3)
%  h: (4,5,1)(5,6,3)
%  i: (8,0,2)(9,1,9)
%  j: (2,1,3)(8,4,9)
%  k: (2,4,4)(3,5,6)
%  l: (0,4,1)(1,11,11)
%  m: (2,0,9)(2,4,11)
%  n: (1,5,4)(2,9,6)

These conditions imply that for $i=1,2,3$, the edges with color $i$
induce a tree $\TT_i$ rooted at $v_i$, where all edges of $\TT_i$ are directed towards the root.
We denote the Schnyder wood by $(\TT_1. \TT_2, \TT_3)$.
Every $3$-connected plane graph $G$ has a Schnyder wood~\cite{FZ08,BF12}.
From a Schnyder wood of
% a $3$-connected plane graph
$G$, one can construct a \df{dual Schnyder wood} (the Schnyder wood for the
dual of $G$). Consider the dual graph $G^*$ of $G$ in which the vertex
for the outer face of $G$ has been split into three vertices forming a triangle. These three vertices
represent the three regions between pairs of half edges from the outer vertices of $G$. Then a Schnyder
wood for $G^*$ is formed by orienting and coloring the edges so that between an edge $e$ in $G$
and its dual $e^*$ in $G^*$, all three colors $1,2,3$ are used. In particular, if $e$ is uni-directional
in color $i$, $i=1,2,3$, then $e^*$ is bi-directional in colors $i-1$, $i+1$ and vice versa;
see Fig.~\ref{fig:can-schA}.

\subsection{Correspondence}
\label{subsect:corr}

Let $G$ be a $3$-connected plane graph with a specified base-pair $(v_1,v_2)$ and a specified head
vertex $v_3$ such that $v_1$, $v_2$, $v_3$ are in that counterclockwise order on the outer face.
It is known that an ordered path partition of $G$ defines a Schnyder wood on $G$, where the three outgoing
edges for each vertex are to its (1)~leftmost predecessor, (2)~rightmost predecessor, and
(3)~highest-level successor~\cite{FZ08,BF12}. We call an ordered path partition and the corresponding Schnyder
wood computed this way to be \df{compatible} with each other.
Badent~et~al.~\cite{BBC11} claim that the converse can also be done, that is, given a Schnyder wood on $G$,
one can compute an ordered path partition, compatible with the Schnyder wood (and hence,
there is a one-to-one correspondence between the concepts).
However, it turns out that the algorithm in~\cite{BBC11} for converting a Schnyder wood to a compatible ordered path partition is incomplete\footnote{Confirmed by a personal communication with the authors of~\cite{BBC11}.}, that is, the computed ordered path partition is not always
compatible with the Schnyder wood; see Fig.~\ref{fig:non-eq}
% in Appendix~\ref{sec:proofs}
 for an example. Here we provide a correction for the algorithm.

\begin{figure}[t]
\centering
\includegraphics[width=\textwidth]{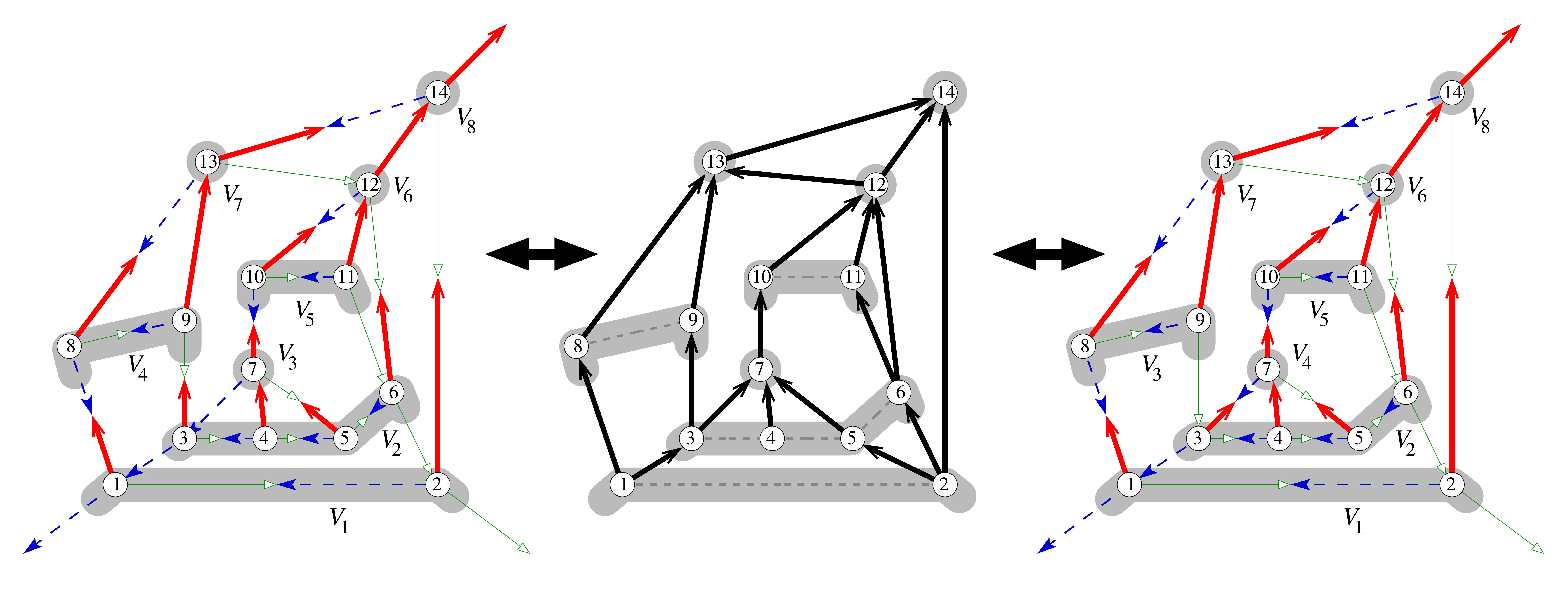}
\caption{Two ordered path partition with their compatible Schnyder woods (left and right) that gives the
 same acyclic graph $D_i$ after reversing the direction of the edges in two trees and grouping all cyclic
 maximal paths. Thus starting with either Schnyder wood can give either ordered path partition by the
 algorithm in~\cite{BBC11}.}
\label{fig:non-eq}
\end{figure}

Let $(\TT_1, \TT_2, \TT_3)$ be a Schnyder wood of $G$. From here on whenever
we talk about Schnyder woods, we consider a circular order for the indices $i=1,2,3$ so that
$(i-1)$ and $(i+1)$ are well defined when $i\in \{1,2,3\}$. By~\cite{FZ08}
there is no cycle of $G$ that is directed in $\TT_{i-1}^{-1}\cup \TT_{i+1}^{-1}\cup \TT_{i}$, where
$\TT_i^{-1}$, $i=1,2,3$, denotes the reversed edges of $\TT_i$.
Since there are some bi-directional edges in $G$ colored
with $i+1$ and $i-1$ (which we call \df{cyclic}), $\TT_{i-1}^{-1}\cup \TT_{i+1}^{-1}\cup \TT_{i}$ induces some directed cycles of length $2$. We can form a directed acyclic graph $D_i$ by \textit{grouping}
each maximal path in $\TT_{i-1}^{-1}\cup \TT_{i+1}^{-1}\cup \TT_{i}$ with cyclic bi-colored edges
(call such a path a \textit{cyclic maximal path}) into a single vertex. Here the maximal paths $P_i$ with
cyclic bi-colored edges are the vertices of $D_i$, and for
two such paths $P_i$ and $P_j$, there is a directed edge from $P_i$ to $P_j$ in $D_i$ whenever
there is a directed (not cyclic) edge $(u,v)$ in $\TT_{i-1}^{-1}\cup \TT_{i+1}^{-1}\cup \TT_{i}$
for some $u\in P_i$ and $v\in P_j$.
Badent~et~al.~\cite{BBC11} showed that $D_i$ is acyclic and they suggest
to obtain an ordered path partition by taking a topological order of $D_i$.
However, the resulting ordered path partition is not necessarily compatible with $(\TT_1, \TT_2, \TT_3)$ and Fig.~\ref{fig:non-eq} shows an example.
Instead, before grouping the cyclic maximal paths, we augment the graph
$\TT_{i-1}^{-1}\cup \TT_{i+1}^{-1}\cup \TT_{i}$ with the following directed edges.
For each vertex $v$ of $G$, we add a directed edge from each child of $v$ in $\TT_{i-1}$
and $\TT_{i+1}$ to the parent of $v$ in $\TT_i$. The augmented graph remains acyclic
and it is consistent with the partial order defined by $D_i$ (we call this the \textit{partial order defined by
$(\TT_{i-1}^{-1}\cup \TT_{i+1}^{-1}\cup \TT_{i})$}).
A topological order of the augmented graph (after grouping all cyclic maximal paths)
induces a compatible ordered path partition.
% Further details and the complete proof of the Lemma below are provided in Appendix~\ref{sec:proofs}.

%\wormhole{lm-schny-can}
\begin{lemma}
\label{lem:schny-can}
Let $(\TT_1, \TT_2, \TT_3)$ be a Schnyder wood of a $3$-connected plane graph
$G$ with three specified vertices $v_1$, $v_2$, $v_3$ in that counterclockwise order on the outer face.
Then for $i=1,2,3$ one can compute in linear time an ordered path partition $\Pi_i$ compatible with
$(\TT_1, \TT_2, \TT_3)$ such that $\Pi_i$ has $(v_{i-1}, v_{i+1})$ as the base pair and $v_i$ as the head.
The ordered path partition is consistent with the partial order defined by
%the directed acyclic graph $D_i$ obtained by grouping cyclic maximal paths in
$(\TT_{i-1}^{-1}\cup \TT_{i+1}^{-1}\cup \TT_{i})$.
\end{lemma}
\begin{proof}
Consider the directed acyclic graph $D_i$ by grouping each cyclic maximal path of
 $(\TT_{i-1}^{-1}\cup \TT_{i+1}^{-1}\cup \TT_{i})$ into a single vertex.
% In particular, the vertices of $D^*_i$ are the
% maximal undirected paths in $D_i$ and for two vertices $U,V$ of $D^*_i$, there is a directed edge
% $(U,V)$ if and only if there exists some vertices $u\in U$ and $v\in V$ such that $(u,v)$ is a directed
% edge in $D^*$. The graph $D^*_i$ remains acyclic since by~\cite{FZ08}, no cycle in $G$ is directed
 $D_i$ is acyclic since by~\cite{FZ08}, no cycle in $G$ is directed in
 $\TT_{i-1}^{-1}\cup \TT_{i+1}^{-1}\cup \TT_{i}$.
 Furthermore each vertex in $D_i$ has at least two predecessors, one in $\TT_{i-1}^{-1}$ and one in
 $\TT_{i+1}^{-1}$. Therefore one can compute an ordered path partition of $G$ by taking a topological
 ordering $\Pi$ of $D_i$ and for each vertex $u$ of $G$ assigning $u$ label $k$ where $u\in U$ and
 $k$ is the rank of $U$ in $\Pi$; see~\cite{BBC11} for details. However the ordered path partition obtained
 by this procedure might not be compatible with $(\TT_1, \TT_2, \TT_3)$; in particular for a vertex $u$ of
 $G$, its parent in $T_i$ might not be its highest-level successor; see Fig.~\ref{fig:non-eq}.

 In order to ensure compatibility between
 the Schnyder wood and the obtained ordered path partition, we further augment $D_i$ by adding some
 extra edges. In particular, for each vertex $u$ of $G$, if $v$ is its parent in $\TT_i$ and $w$ is it child in
 either $\TT_{i-1}$ or $\TT_{i+1}$, then we add a directed edge $(V,W)$ in $D^*_i$ where $V$, $W$ are
 cyclic maximal paths in $D_i$ and $v\in V$, $w\in W$. Call the augmented directed graph $H_i$.
 We now show that with the addition of the extra edges the directed graph $H_i$ remains acyclic.
 We prove this claim only for $H_3$; for $H_1$ and $H_2$ the proofs are analogous.

\begin{figure}[t]
\centering
\includegraphics[width=0.9\textwidth]{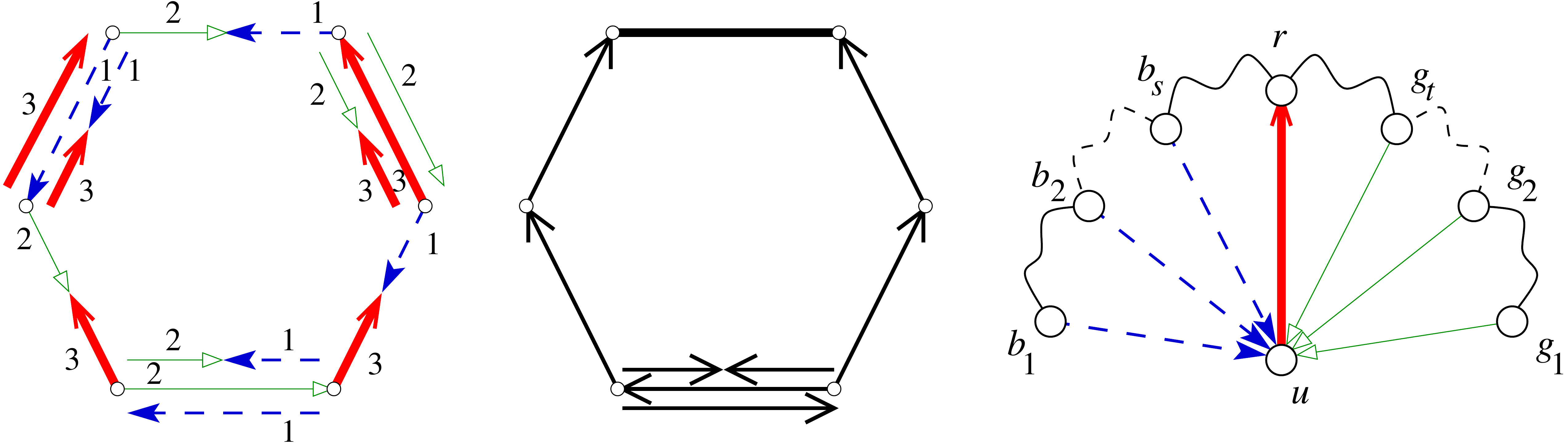}\\
(a)\hspace{0.27\textwidth}(b)\hspace{0.27\textwidth}(c)\hspace{0.05\textwidth}
\caption{(a) Coloring and orientation of the edges around a generic face in to a Schnyder wood
 $(\TT_1, \TT_2, \TT_3)$, (b) orientation of the edges in
 $D_3=(\TT_{2}^{-1}\cup \TT_{1}^{-1}\cup \TT_{3})\setminus\TT_{2,1}$, (c) directed paths in
 $\TT_{2}^{-1}\cup \TT_{1}^{-1}\cup \TT_{3}$ from children of a vertex $u$ in $\TT_1$ and $\TT_2$
 to the parent in $\TT_3$.}
\label{fig:face-schny}
\end{figure}

Consider an arbitrary face $f$ of $G$. By~\cite{FZ08}, the edges on the boundary of $f$ can be
 partitioned into at most six consecutive sets in clockwise order around $f$
 (see Fig.~\ref{fig:face-schny}(a)):

\begin{enumerate}[(i)]
	\item one edge from the set {clockwise in color $1$, counterclockwise in color $2$, bi-colored with
		a clockwise $1$ and counterclockwise $2$}
	\item zero or more edges bi-colored in counterclockwise $2$ and clockwise $3$
	\item one edge from the set {clockwise in color $3$, counterclockwise in color $1$, bi-colored with
		a clockwise $2$ and counterclockwise $1$}
	\item zero or more edges bi-colored in counterclockwise $1$ and clockwise $2$
	\item one edge from the set {clockwise in color $2$, counterclockwise in color $3$, bi-colored with
		a clockwise $2$ and counterclockwise $3$}
	\item zero or more edges bi-colored in counterclockwise $3$ and clockwise $1$

\end{enumerate}

Fig.~\ref{fig:face-schny}(b) shows the direction of the edges of $f$ in $D_3$. Now consider a vertex $u$
 of $G$. Let $r$ be its parent in $\TT_3$, $b_1$, $\ldots$, $b_s$ its children in $\TT_1$ in clockwise order
 around $u$, and $g_1$, $\ldots$, $g_t$ its children in $\TT_2$ in counterclockwise order around $u$.
 Thus $L=b_1, \ldots, b_s, r, g_t, \ldots, g_1$ appears consecutively around $u$ in that clockwise order.
 Let $P^b_i$, $i=1, \ldots, s$ be the paths from $b_i$ to its next vertex in $L$ that forms the face of $G$
 containing these two vertices and $u$; see Fig.~\ref{fig:face-schny}(c). Similarly let $P^g_i$, $i=1, \ldots,
 t$ be the paths from $g_i$ to its previous vertex in $L$ that forms the face of $G$ containing these two
 vertices and $u$. From Fig.~\ref{fig:face-schny}(b) one can see that all the paths $P^b_i$ and $P^g_i$
 are directed towards $r$ (possibly with some bi-directional edges) in $D_3$. This imply that for each
 vertex $b_i$, $i=1\ldots, s$ (resp. $g_i$, $i=1,\ldots, t$), there is a directed path (possibly with some
 bi-directional edges) from $b_i$ (resp. $g_i$) to $r$. Thus adding an edge from any $b_i$ or $g_i$ to $r$
 does not create any cycle in $D_3$ since otherwise replacing the edge with the directed path induces
 a directed cycle in $\TT_{2}^{-1}\cup \TT_{1}^{-1}\cup \TT_{3}$, a contradiction. Since the addition of the
 extra edges does not make any cycle in $D_3$, the graph $H_3$ remains acyclic.

Once we add the extra edges, we can compute an ordered path partition consistent with $D_i$, by taking a
 topological ordering of $H_i$ and for each vertex $u$ of $G$, assigning $k$ as its label
 where $u\in U$ for some cyclic maximal path $U$ in $D_i$ and $k$ is the rank of $U$ in the
 topological ordering of $H_i$. This ordered path partition is compatible with $(\TT_1, \TT_2, \TT_3)$ since
 for each vertex $u$, the parents of $u$ in $\TT_{i-1}$ and in $\TT_{i+1}$ are the leftmost and rightmost
 predecessors (due to the embedding) and the parent in $\TT_i$ is the highest-level successor
 (due to the addition of the extra edges).

The time complexity of the above algorithm is linear. Computing the directed graph $D_i$, $i=1,2,3$
 can be done by traversing the edges of $G$. Addition of the extra
 edges takes $\sum_{v\in V(G)}deg(v)=2|E(G)|=\Oh(n)$ time also. Finally the topological ordering can be
 done by a linear-time Depth-first traversal on the graph $H_i$, which also has a linear size.
\end{proof}

Fig.~\ref{fig:sch-order} illustrates the three ordered path partitions computed from a Schnyder wood of
 a $3$-connected plane graph using the algorithm described above.

\begin{figure}[htbp]
\centering
\includegraphics[width=\textwidth]{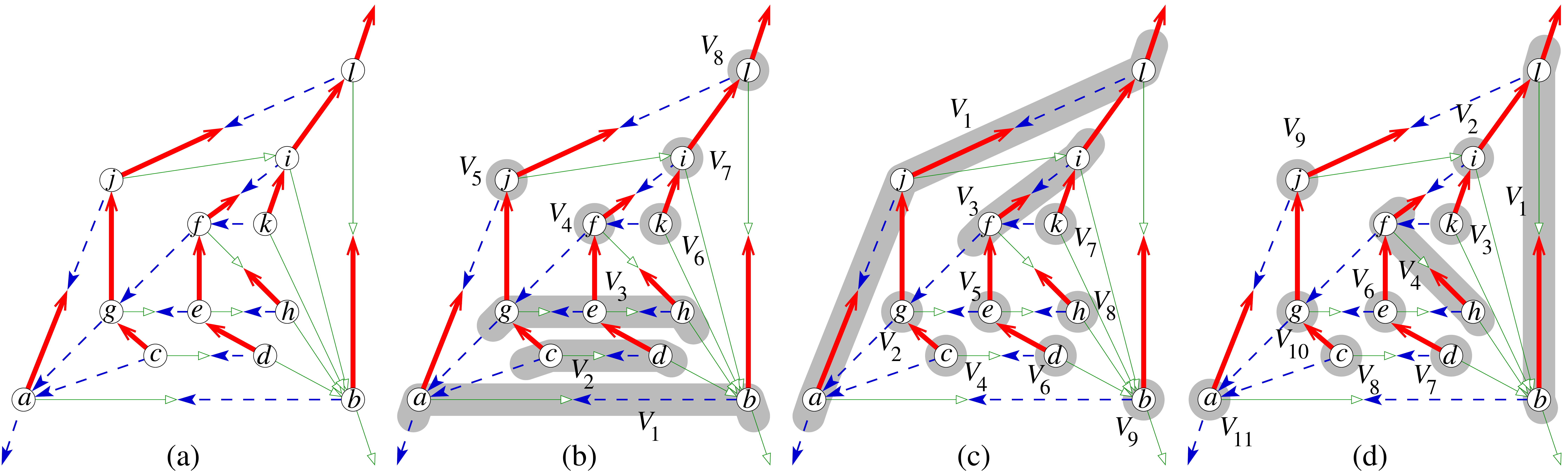}
\caption{(a) A Schnyder wood $(\TT_1, \TT_2, \TT_3)$ in a $3$-connected plane graph $G$, which is not
 compatible with any canonical order, (b)--(d) the three compatible ordered path partition computed from
 $(\TT_1, \TT_2, \TT_3)$ using the algorithm described in the proof of Lemma~\ref{lem:schny-can}.}
\label{fig:sch-order}
\end{figure}

%\begin{observation}
%\label{obs:can-schny} Let $\Pi=(V_1,V_2,\ldots,V_L)$ be an ordered path partition in a $3$-connected plane
% graph $G$ with a specified base-pair $(v_1,v_2)$ and a specified head vertex $v_3$ such that
% $v_1$, $v_2$, $v_3$ are in that counterclockwise order on the outer face. Let  $(\TT_1, \TT_2, \TT_3)$
% be the Schnyder wood compatible with $\Pi$. Then the maximal paths in $(\TT_1, \TT_2, \TT_3)$
% with each edge bi-colored with colors $1$ and $2$ are in one-to-one correspondence with the vertices
% with the same level.
%\end{observation}

%
% Clearly given an ordered path partition with a specified base pair $(v_1, v_2)$ and a specified head
% vertex $v_3$, one can compute in linear time (following Lemma~\ref{lem:can-schny})
% a compatible Schnyder wood with the specified vertices $v_1$, $v_2$, $v_3$.
%In the following we show that the converse can also be done in linear time.
%
%

\subsection{Elementary Canonical Orders and Elementary Schnyder Woods}

Here we introduce the concepts of \df{elementary} canonical orders and \df{elementary} Schnyder woods,
 which we use in the subsequent section for an alternate proof for Theorem~\ref{thm:box}.

A canonical order $\Pi=(V_1,V_2,\ldots,V_L)$ for a $3$-connected plane graph $G$ is \df{elementary}
 if (i) $|V_1|=2$ (or equivalently the base-pair $(v_1, v_2)$ induces an edge on the outer cycle).
%; and (ii) for $2\le k\le L$, whenever $V_k=\{v_a, v_{a+1}, \ldots v_b\}$ consists of more than one vertex
% (i.e., $a<b$), each of $v_a$ and $v_b$ has exactly one neighbor on $C_{k-1}$ and these are the
% only neighbors of $V_k$ in $G_{k-1}$.
 Kant showed that any $3$-connected plane graph $G$, with an edge $(v_1, v_2)$ and a third vertex $v_3$,
 both on the outer cycle, has an elementary canonical order for the base pair $(v_1,v_2)$ and head $v_3$
 and such an elementary canonical order for $G$ can be computed in linear time~\cite{Kan96}.

Not all Schnyder woods are compatible with canonical orders; Fig.~\ref{fig:sch-order}(a) shows an example
 of a Schnyder wood that is not compatible with any canonical order.
 Let $\Pi=(V_1,V_2,\ldots,V_L)$ be a canonical order in a $3$-connected plane graph $G$ with
 a specified base-pair $(v_1,v_2)$ and a specified
 head vertex $v_3$ such that $v_1$, $v_2$, $v_3$ are in that counterclockwise order on the outer face.
 Let  $(\TT_1, \TT_2, \TT_3)$ be the Schnyder wood compatible with $\Pi$.
 Then by the definition of compatible Schnyder wood, it follows that
% from Lemma~\ref{lem:can-schny} and Observation~\ref{obs:can-schny}, it follows that
 for every maximal path $P=\{u_1, u_2, \ldots u_t\}$ in $(\TT_1, \TT_2, \TT_3)$
 with each edge $(u_j, u_{j+1})$ bi-colored in color $1$ and $2$ for $1\le j<t$, there is no child of the vertices
 $u_2, \ldots, u_{t-1}$ in $\TT_3$, and the only children of $u_1$ and $u_t$ in $\TT_3$ (if any) are
 respectively the parent of $u_1$ in $\TT_1$ and the parent of $u_t$ in $\TT_3$.
We call this property for a Schnyder wood $(\TT_1, \TT_2, \TT_3)$ the \df{canonical property} of the
 Schnyder wood for color 3. One can define the canonical property of a Schnyder wood analogously
 for colors $1$ and $2$.
Let $(\TT_1, \TT_2, \TT_3)$ be a Schnyder wood for a $3$-connected plane graph $G$ with three specified
 vertices $v_1$, $v_2$, $v_3$ in that counterclockwise order on the outer face. Then
 $(\TT_1, \TT_2, \TT_3)$ is \textit{elementary} for the color $i$, $i=1,2,3$, if (i) $(v_{i-1},v_{i+1})$ is an
 edge on the outer cycle, and (ii) the canonical property holds for color $i$ in $(\TT_1, \TT_2, \TT_3)$.
%for every maximal path $P=\{u_1, u_2, \ldots u_t\}$ with each edge
% $(u_j, u_{j+1})$ bi-colored in color $i-1$ and $i+1$ for $1\le j<t$, there is no child of the vertices
% $u_2, \ldots, u_{t-1}$ in $\TT_i$, and the only child of $u_1$ and $u_t$ in $\TT_i$ (if any) are
% respectively the parent of $u_1$ in $\TT_{i-1}$ and the parent of $u_t$ in $\TT_{i+1}$.

\begin{lemma}
\label{lem:elementary}
 Let $G$ be a $3$-connected plane graph with three specified vertices $v_1$, $v_2$, $v_3$ in the
 counterclockwise order on the outer face of $G$. Let $\Pi_i$ be an ordered path partition of $G$
 with a specified base-pair $(v_{i-1},v_{i+1})$ and a specified head vertex $v_i$ for $i=1,2,3$ and let
 $(\TT_1, \TT_2, \TT_3)$ be the compatible Schnyder wood for $\Pi$. Then
 $\Pi_i$ is an elementary canonical order if and only if $(\TT_1, \TT_2, \TT_3)$ is an elementary
 Schnyder wood for color $i$.
\end{lemma}
\begin{proof} The lemma follows from the definitions of elementary canonical order and elementary
 Schnyder wood and from the observation that the maximal paths in $(\TT_1, \TT_2, \TT_3)$
 bi-colored with colors $i-1$ and $i+1$ are in one-to-one correspondence with the paths of $G$ with
 the same level in $\Pi_i$, $i=1,2,3$.

\end{proof}

We end this section with the following lemma.

\begin{lemma}
\label{lem:dual} If a Schnyder wood $(\TT_1, \TT_2, \TT_3)$ for a $3$-connected plane graph $G$ is
 elementary for color $i$, $i=1,2,3$, then its dual Schnyder wood is also elementary for color $i$
 in the dual graph of $G$.
\end{lemma}
\begin{proof} We prove this lemma only for color $3$, since the proof for color $1$ or $2$ is similar.
By definition the outer cycle for the dual graph of $G$ for constructing the dual Schnyder
 wood is a triangle; hence the base-pair for the dual Schnyder wood are adjacent. Hence for the dual
 Schnyder wood to be elementary it is sufficient that for each $P=\{u_1, u_2, \ldots u_t\}$
 with each edge $(u_j, u_{j+1})$ bi-colored in $1$ and $2$ for $1\le j<t$, (i) there is no child of the vertices
 $u_2, \ldots, u_{t-1}$ in $\TT_3$, and (ii) the only child of $u_1$ and $u_t$ in $\TT_3$ (if any) are
 respectively the parent of $u_1$ in $\TT_1$ and the parent of $u_t$ in $\TT_2$. However, the way
 we define the color and orientation of the edges in the dual, failure to satisfy Condition (i) for some
 path in the dual Schnyder wood implies that Condition (ii) does not hold for some path in the primal
 Schnyder wood and vice versa; see Fig.~\ref{fig:elementary-dual}. Therefore if the primal Schnyder
 wood is elementary in color 3, so is the dual Schnyder wood.
\end{proof}

\begin{figure}[htbp]
\centering
\includegraphics[width=0.37\textwidth]{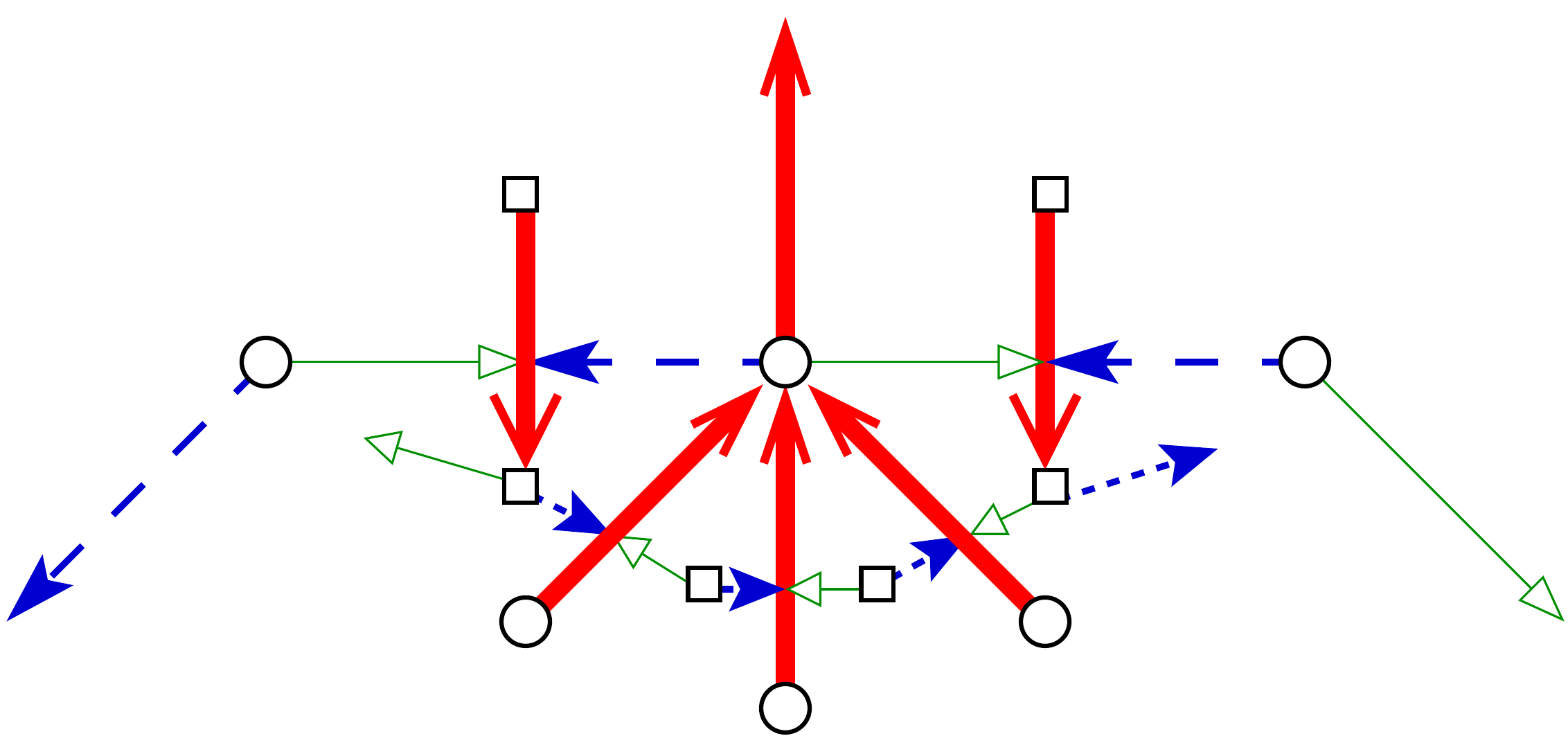}
\caption{Illustration for the proof of Lemma~\ref{lem:dual}.}
\label{fig:elementary-dual}
\end{figure}

\subsection{Orthogonal Surfaces}
\label{sec:ortho}

Here we briefly review the notion of orthogonal surfaces, which we use in the proof for
 Theorem~\ref{thm:box}; see~\cite{FZ08} for more details.
A point $p$ in $\R^3$ \df{dominates} another point $q$ if the coordinate of $p$ is greater than or equal to
 $q$ in each dimension; $p$ and $q$ are \df{incomparable} if neither of $p$ nor $q$ dominates the other.
 Given a set $M$ of incomparable points, an \df{orthogonal surface} defined by $M$ is the geometric
 boundary of the set of points that dominate at least one point of $M$.
% For two points $p$, $q$, their \df{join} (\df{meet}), denoted by $p\lor q$ ($p\land q$), is the point
% obtained by taking the maximum (minimum) coordinate of $p$, $q$ in each dimension separately.
 For two points $p$ and $q$, their \df{join}, $p\lor q$ is obtained by taking the maximum coordinate of
 $p$, $q$ in each dimension separately. The \df{minimums} (\df{maximums})
 of an orthogonal surface $S$ are the points of $S$ that dominate (are dominated by) no other
 point of $S$. An orthogonal surface $S$ is \df{rigid} if for each pair of points $p$ and $q$ of $M$ such that
 $p\lor q$ is on $S$, $p\lor q$ does not dominate any point other than $p$ and $q$.
 An orthogonal surface is \df{axial} if it has exactly three unbounded orthogonal arcs.
 Rigid axial orthogonal surfaces are known to be in one-to-one correspondence with Schnyder woods
 of 3-connected plane graphs~\cite{FZ08} and the rigid axial orthogonal surfaces $S$ and $S^*$
 corresponding to a Schnyder wood and its dual coincide with each other, where the maximums of
 $S$ are the minimums of $S^*$ and vice versa.

\section{Primal-Dual Representations of 3-Connected Planar Graphs}
\label{sec:pri-du}

Here we prove Theorem~\ref{thm:box}.

\begin{backInTime}{thm-box}
 \begin{theorem}
  Every $3$-connected planar graph $G=(V,E)$ admits a proper primal-dual box-contact representation in 3D and it can be computed in $\Oh(|V|)$ time.
 \end{theorem}
\end{backInTime}

%Specifically, we describe a linear-time algorithm that computes
Specifically, we describe two different linear-time algorithms that compute
 a box-contact representation for the primal graph and the dual graph separately
 and then fits them together to obtain the desired result.
%We compute the coordinates of the boxes based on a Schnyder wood. This guarantees that
In the first algorithm we compute the coordinates of the boxes based on a Schnyder wood,
 in the second, we compute them based on an elementary canonical order. Both these algorithms
 guarantee that the boundary for the primal (dual) representation induces an \df{orthogonal surface}
 compatible with the dual (primal) Schnyder wood.
Since the orthogonal surfaces for a Schnyder wood and its dual coincide topologically, we can
fit together the primal and the dual box-contact representation to obtain a desired representation.
%One can visualize the construction as follows. Start with the orthogonal
%surface $S$ corresponding to both a Schnyder wood of a $3$-connected plane graph $G$ and its dual.
% This orthogonal surface $S$ creates two half-spaces in either side of $S$, which we call the
%\df{primal space} and \df{dual space}.
%In both the primal and the dual spaces, the task is to extend boxes from $S$ so that
%(i)~the boxes are interior-disjoint and (ii)~they induce necessary proper contacts.

%The second algorithm (see Appendix) is based on a canonical order of $G$. The algorithm is
%iterative, it adds one box at a time,
%maintaining several construction invariants.

To avoid confusion, we denote a connected region in a plane embedding of a graph
 by a \df{face}, and a side of a 3D shape by a \df{facet}. For a 3D box $R$,
 call the facet with highest (lowest) $x$-coordinate as the $x^+$-facet ($x^-$-facet) of $R$.
The $y^+$-facet, $y^-$-facet, $z^+$-facet and $z^-$-facets of $R$ are defined similarly. For
 convenience, we sometimes denote the $x^+$-, $x^-$, $y^+$-, $y^-$, $z^+$- and $z^-$-facets
 of $R$ as the \df{right}, \df{left}, \df{front}, \df{back}, \df{top} and \df{bottom} facets of $R$, respectively.

%\begin{backInTime}{thm-box}
% \begin{theorem}
%  Every $3$-connected planar graph $G=(V,E)$ admits a proper primal-dual box-contact representation in 3D and it can be computed in $\Oh(|V|)$ time.
% \end{theorem}
%\end{backInTime}

\begin{figure}[tb]
%\vspace{-0.2cm}
\centering
\parbox[t][4.5cm][t]{0.07\textwidth}
{
	\includegraphics[width=0.28\textwidth]{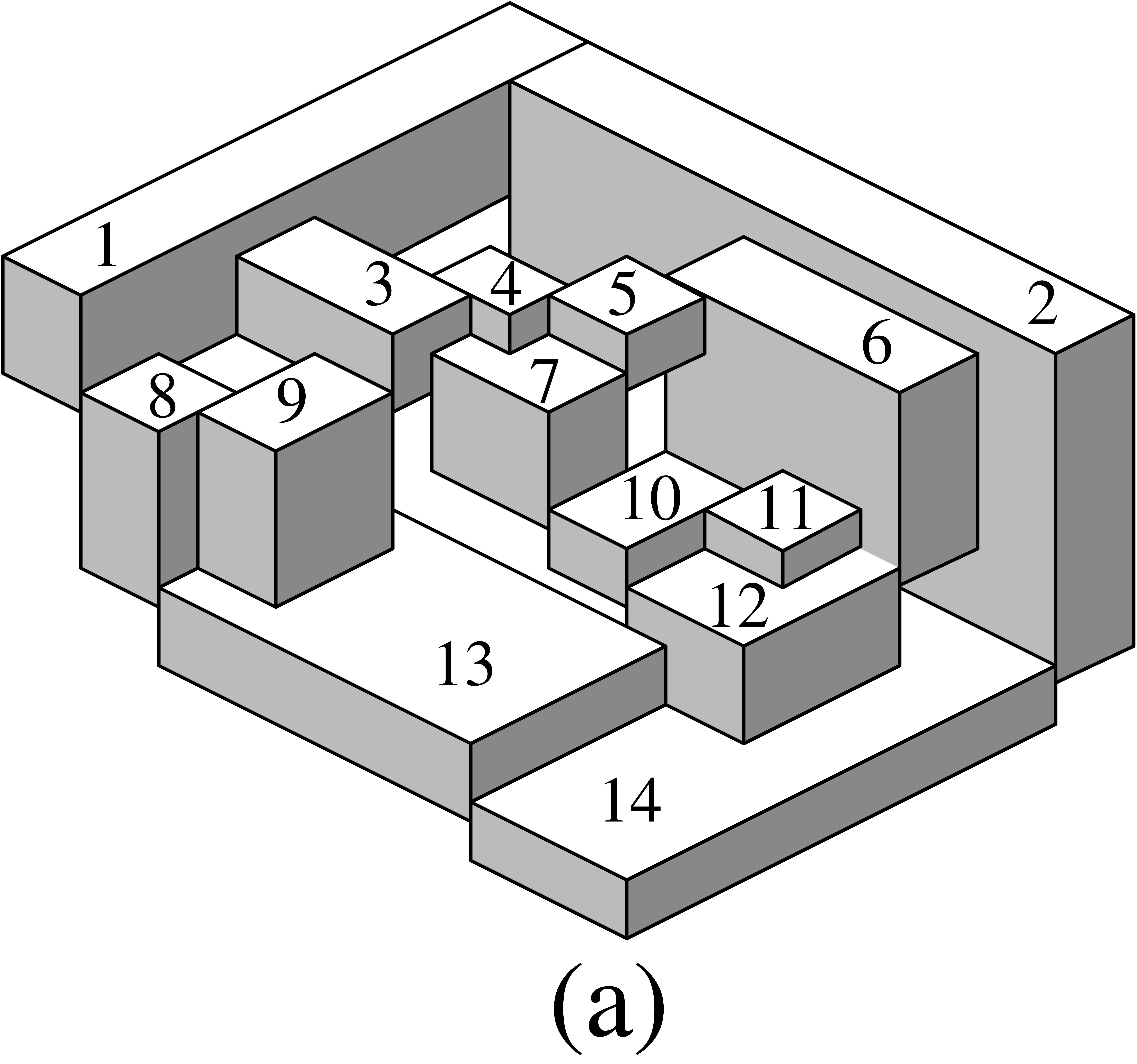}
}
\parbox[t][4.5cm][b]{0.62\textwidth}
{
	\includegraphics[width=0.85\textwidth]{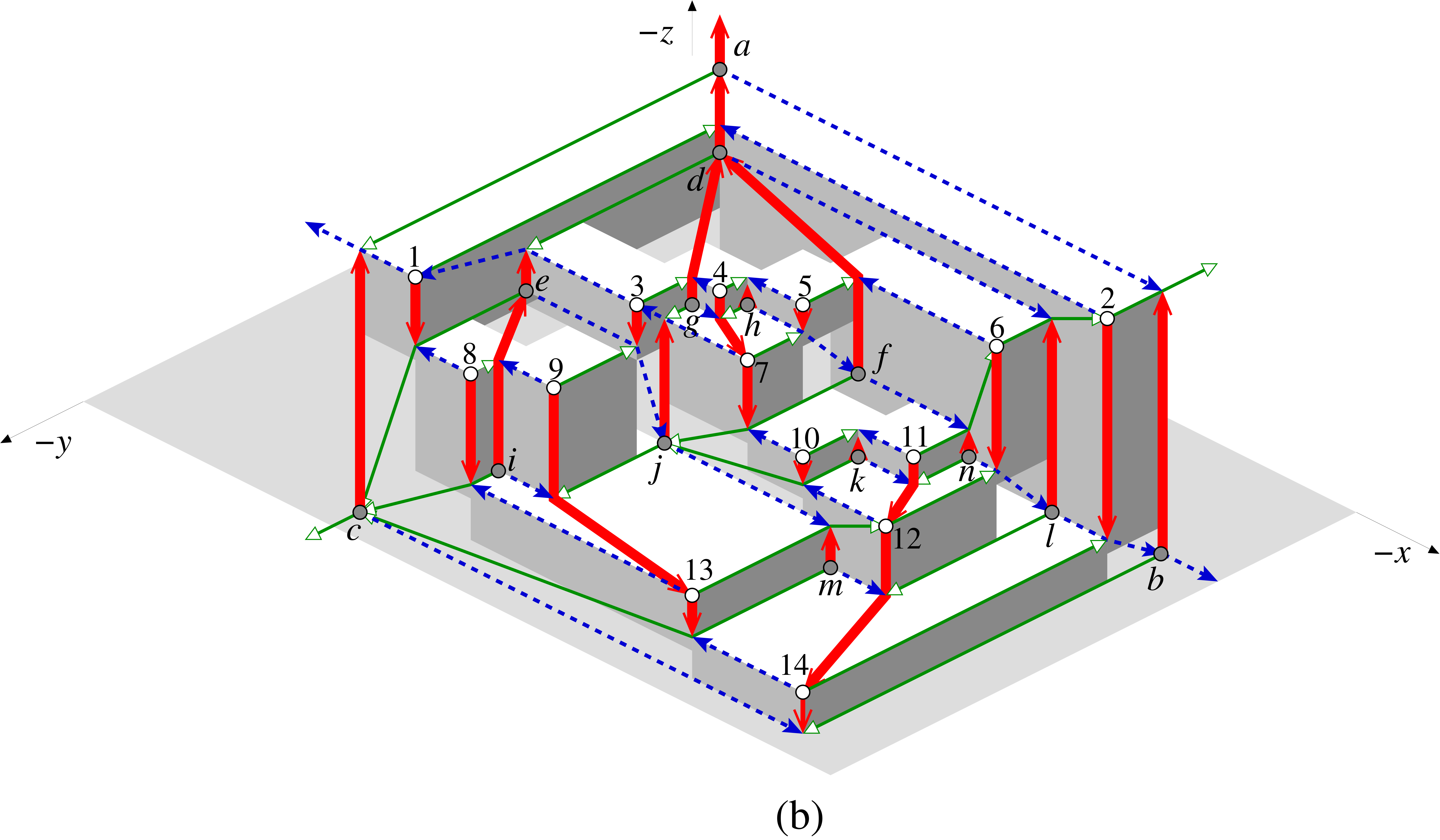}
}
\parbox[t][4.5cm][t]{0.28\textwidth}
{
	\includegraphics[width=0.28\textwidth]{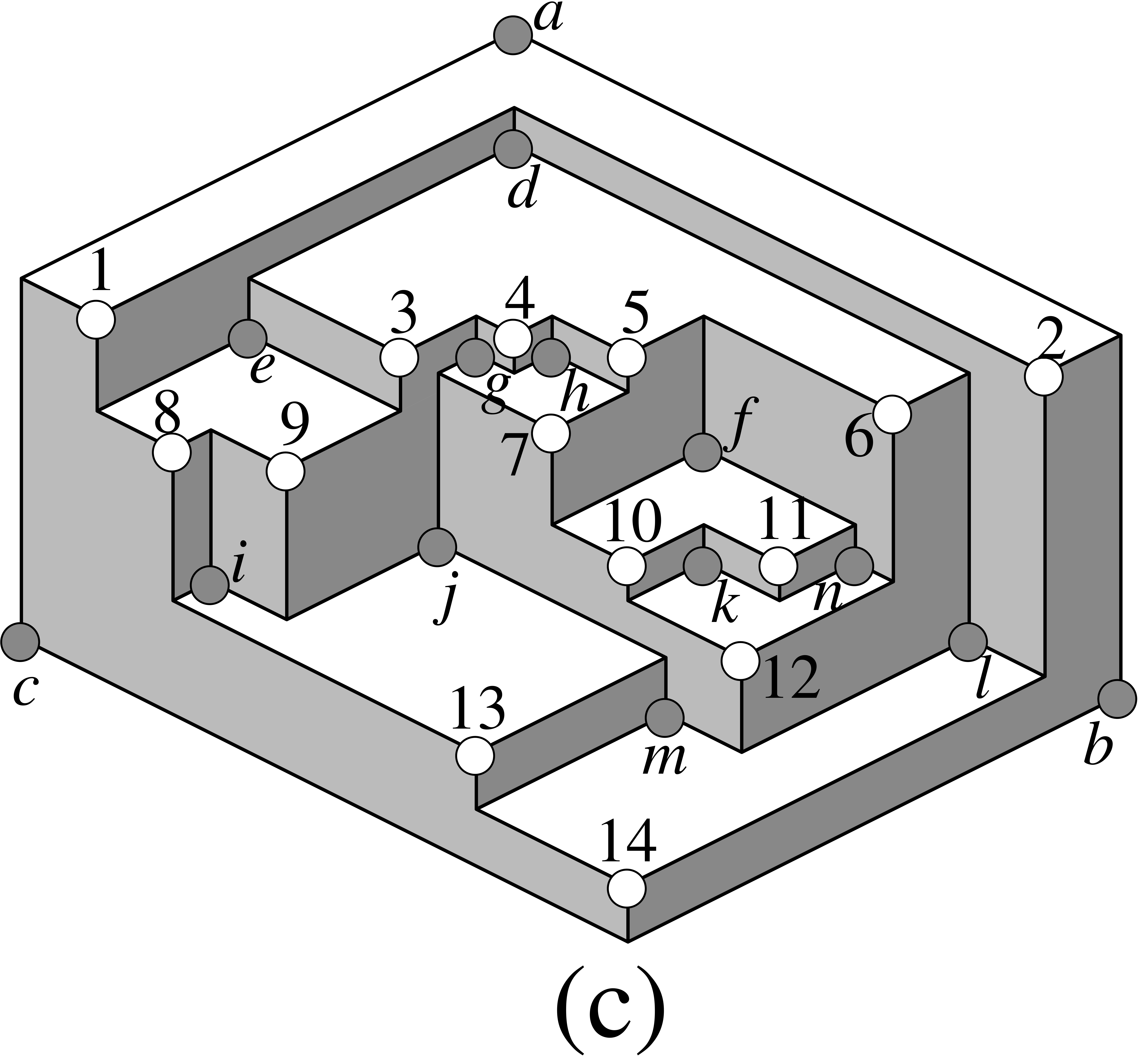}
}
    \caption{Box-contact representation (a) for the graph in Fig.~\ref{fig:can-sch}
	with its primal-dual Schnyder wood (b) and the associated orthogonal surface (c).
	The thick solid red, dotted blue and thin solid green edges represent the three trees
	in the Schnyder wood.}
    \label{fig:draw}
\end{figure}

\subsection{Construction Based on Schnyder Woods}
\label{sec:mainproof}

%\begin{proof}
We first construct a contact representation $\Gamma$ of the primal graph $G$ with boxes in 3D.
 Let $v_1$, $v_2$ and $v_3$ be three vertices on the outer face of $G$ in counterclockwise order.
 We compute a Schnyder wood $(\TT_1, \TT_2, \TT_3)$ such that for $i=1,2,3$, $\TT_i$ is rooted
 at $v_i$. By Lemma~\ref{lem:schny-can}, for $i=1,2,3$, one can compute a compatible ordered path
 partition with the base-pair $(v_{i-1}, v_{i+1})$ and head $v_i$, which is consistent with the partial order
 defined by $(\TT_{i-1}^{-1}\cup \TT_{i+1}^{-1}\cup \TT_{i})$. Denote by
 $<_X$, $<_Y$ and $<_Z$ the three ordered path partitions compatible with $(\TT_1, \TT_2, \TT_3)$, that
 are consistent with $D_1$, $D_2$, and $D_3$, respectively. We use these three ordered path partitions
 to define our box-contact representation for $G$.

 For a vertex $u$, let $x_M (u)$, $y_M (u)$, and $z_M (u)$ be the levels of $u$ in the ordered path
 partitions $<_X$, $ <_Y$, and $<_Z$, respectively. Define $x_m(u)=x_M(b)$, $y_m(u)=y_M(g)$ and
 $z_m(u)=z_M(r)$, where $b$, $g$ and $r$ are the parents of $u$ in $\TT_1$, $\TT_2$ and $\TT_3$,
 respectively, whenever these parents are defined. For each of the three special vertices $v_i$,
 $i=1,2,3$, the parent is not defined in $\TT_i$. We assign $x_m(v_1)=0$, $y_m(v_2)=0$ and
 $z_m(v_3)=0$.

Now for each vertex $u$ of $G$, define a box $R(u)$ representing $u$ as the region
 $[x_M(u), x_m(u)]\times[y_M(u), y_m(u)]\times[z_M(u), z_m(u)]$.
Denote by $R$ the set of all boxes $R(u)$ for the vertices $u$ of $G$. We show in
Lemma~\ref{lem:box-schny} that $R$ yields a box-contact representation for $G$.
Furthermore, for each edge $(u,v)$ of $G$, if $(u,v)$ is uni-colored then $R(u)$ and $R(v)$
make a proper contact. Otherwise, assume \WLOG that $(u,v)$ is bi-colored
with colors $1$ (oriented from $u$ to $v$) and $2$ (oriented from $v$ to $u$).
Then $x_m(u)=x_M(v)$, $y_M(u)=y_m(v)$, $z_M(u)=z_M(v)$;  $R(u)$ and $R(v)$
make a non-proper contact along the line-segment $\{x_m(u)\}\times\{y_m(v)\}\times[z_M(u), z^*]$,
where $z^*=\min\{z_m(u),z_m(v)\}$. However since $v$ is the only parent of $u$ in $\TT_1$
and $u$ is the only parent of $v$ in $\TT_2$, by Lemma~\ref{lem:box-schny},
the $x^+$-facet of $R(u)$ and the $y^+$-facet of $R(v)$ do not make a proper contact with any box.
Hence either extending $R(u)$ in the positive $x$-direction or extending $R(v)$ in positive $y$-direction
by some small amount $0<\epsilon<1$ makes the contact between $R(u)$ and $R(v)$ proper
without creating any overlap or unnecessary contacts between the boxes.
 We thus obtain a proper box-contact representation $\Gamma$ for the primal graph $G$.

Now we describe how to construct the box-contact representation for the dual of $G$.
Consider the orthogonal surface induced by $\Gamma$. Each vertex $v$ of $G$ corresponds
 to the $x^-y^-z^-$-corner $p$ of the box for $v$. The three outgoing edges of $v$ in the Schnyder
 wood can be drawn on the surface from $p$ in the directions $x^+$, $y^+$, $z^+$; see Fig.~\ref{fig:draw}.
Each face of $G$ corresponds to a reflex corner of the orthogonal surface and there is a similar
 (opposite) direction for the outgoing edges in the dual Schnyder wood. The topology of this
 (rigid axial) orthogonal
 surface is uniquely defined by both the Schnyder wood and its dual~\cite{FZ08}.
Thus, we construct the contact representation $\Gamma'$ for the dual of $G$, then (after possible scaling)
 the boundary of $\Gamma'$ exactly matches $B-\Gamma$, where $B$ is the bounding box of $\Gamma$.
We fit $\Gamma$ and $\Gamma'$ together by replacing the three boxes for the three outer vertices of $G$ with a single shell-box, which forms the boundary of the entire representation.

It is easy to see that the above algorithm runs in $\Oh(|V|)$ time. A Schnyder wood for $G$ and the dual Schnyder wood
 for the dual of $G$ can be computed in linear time~\cite{FZ08}. For both the primal graph and the dual
 graph, the three ordered path partitions can then be computed in linear time from these Schnyder woods,
 due to Lemma~\ref{lem:schny-can}. The coordinates of the boxes are then directly assigned in constant
 time per vertex. Finally the primal and the dual representation can be combined together by reflecting
 (and possibly scaling) the dual representation, which can also be done in linear time.
%\end{proof}

Lemma~\ref{lem:box-schny} shows the correctness
% for the construction in the proof of Theorem~\ref{thm:box}; specifically,
 for the above construction; specifically,
 it shows that the boxes computed in the construction induce a contact representation for $G$.

\begin{lemma}
\label{lem:box-schny}
The set of all boxes $R(u)$ for the vertices $u$ of $G$ induces a contact representation of $G$, where
 for each vertex $u$, the $x^+$-, $y^+$-, $z^+$-facets of $R(u)$ touch the boxes for the parents of
 $u$ in $\TT_1$, $\TT_2$, $\TT_3$, respectively, and the $x^-$-, $y^-$-, $z^-$-facets of $R(u)$ touch
 the boxes for the children of $u$ in $\TT_1$, $\TT_2$, $\TT_3$, respectively. Furthermore for each
 edge $(u,v)$ of $G$, the contact between $R(u)$ and $R(v)$ is proper if and only if $(u,v)$ is uni-colored.
\end{lemma}
\begin{proof} We prove the lemma by showing the following two claims:
 (i)~for each edge $(u,v)$ of $G$, the two boxes $R(u)$ and $R(v)$ make contact in the specified facets,
 (ii)~for any two vertices $u$ and $v$ of $G$, the two boxes $R(u)$ and $R(v)$ are interior-disjoint.

\begin{enumerate}[(i)]
%\vspace{-0.2cm}
\item Take an edge $(u,v)$ of $G$. If $(u,v)$ is uni-colored in $G$, assume \WLOG that it has color 3 and is directed from $u$ to $v$.
 By construction, $z_m(u)=z_M(v)$. We now show that $x_M(v)<x_M(u)<x_m(v)$ and
 $y_M(v)<y_M(u)<y_m(v)$, which implies the correct contact between $B(u)$ and $B(v)$.
 Since $(u,v)$ is in $\TT_3$ and the ordered path partitions $<_X$ and $<_Y$ are consistent
 with $\TT_3^{-1}$, $x_M(v)<x_M(u)$ and $y_M(v)<y_M(u)$.
 To show that $x_M(u)<x_m(v)$, consider the parent $b$ of $v$ in $\TT_1$.
 Then $x_m(v)=x_M(b)$. Furthermore, before computing the topological order to find $<_X$, we added
 the directed edge $(u,b)$; thus $x_M(u)<x_M(b)=x_m(v)$. The proof that $y_M(u)<y_m(v)$ is symmetric.

 If $(u,v)$ is bi-colored with colors, \WLOG, $1$ (orientated from $u$ to $v$) and $2$ (from $v$ to $u$),
 then by construction $x_m(u)=x_M(v)$ and $y_M(u)=y_m(v)$. 
 $(u,v)$ is bi-colored in $1$, $2$, so $z_M(u)=z_M(v)$. Thus, $R(u)$, $R(v)$ make
 non-proper contact in the correct facets.

\item Now we show that for any two vertices $u$ and $v$ of $G$, $R(u)$ and $R(v)$ are
 interior-disjoint. By the properties of Schnyder woods, from any vertex $u$ of $G$,
 there are three vertex-disjoint paths $P_1(u)$, $P_2(u)$, $P_3(u)$ where $P_i(u)$ is 
 a directed path from $u$ to $v_i$ of edges colored $i$, $i=1,2,3$.
We first claim that for some $i,j\in\{1,2,3\}$, $i\neq j$, there is a directed path from $u$ to $v$,
 or from $v$ to $u$ in $\TT_{i}\cup \TT_{j}^{-1}$. The claim holds trivially if $u$
 is on the directed path $P_i(v)$, or $v$ is on the directed path $P_i(u)$, for some $i\in\{1,2,3\}$.
 Otherwise, assume that $v$ is in the region between $P_2(u)$, and $P_3(u)$. Then the path $P_1(v)$
 intersects either $P_2(u)$, and $P_3(u)$ at some vertex $w$. Assume \WLOG that
 $w$ is on $P_2(u)$. Then the path $P$ that follows $P_1(v)$ from $v$ to $w$, and then follows
 $P_2^{-1}(u)$ from $w$ to $u$ is directed in $\TT_2^{-1}\cup\TT_1$ (here $P_2^{-1}(u)$ is the path
 $P_2(u)$ with all the directions reversed). Hence, assume that there is a path from $u$ to $v$ in
 $\TT_1\cup\TT_2^{-1}$. By definition, $x_m(u)\le x_M(v)$ and thus $R(u)$ and $R(v)$ are
 interior-disjoint.
\end{enumerate}\vspace{-0.7cm}
\end{proof}

\subsection{Construction Based on Elementary Canonical Orders}
\label{sec:alternateproof}

Here we provide an alternate proof for Theorem~\ref{thm:box} by a different construction algorithm
 than the one in Section~\ref{sec:mainproof}.
This algorithm is based on an elementary canonical order and builds a representation iteratively
 inserting boxes for the vertices in this order. It is similar to the box-contact representation algorithm
 suggested in~\cite{BEFHK+12} for maximal planar graphs. However in addition to generalizing the
 construction for $3$-connected plane graphs, we maintain a stronger invariant on every iteration,
 in order to accommodate the boxes of the dual graph later on.

%\begin{backInTime}{thm-box}
%\begin{theorem}
%Every $3$-connected planar graph $G=(V,E)$ admits a proper primal-dual box-contact representation
% in 3D and it can be computed in $\Oh(|V|)$ time.
%\end{theorem}
%\end{backInTime}

%\begin{proof}
We first construct a contact representation $\Gamma$ of $G$ with boxes.
Let $v_1$, $v_2$ and $v_3$ be three vertices on the outer face of $G$ in counterclockwise order
 such that $(v_1, v_2)$ is an edge on the outer cycle. Compute an elementary canonical
 order $\Pi=(V_1,\ldots,V_L)$ and the compatible Schnyder wood defined by $\Pi$.
Let $G_k$ be the graph induced by the vertices in $V_1\ldots V_k$. We now add boxes
 representing the vertices in the order defined by $\Pi$.

For step $1$ of our construction, we add two boxes to represent $v_1$ and $v_2$
 such that the $x^-$-facet of the box for $v_1$ touches the $x^+$-facet of the box for $v_2$.
Their $z^-$-facets lie on the plane $z=0$ and their $z^+$-facets on $z=1$.
We maintain the invariant that, at the beginning of step $k$,
the boxes of \df{active vertices}, that
is, the vertices in $C_{k-1}$,
intersect the plane $z=k$ in a \df{staircase shape}, that is, where
the minimum $x,y$ boundary of the intersection is an $x,y$-monotone
axis-aligned polyline where the convex corners are $x^-y^-z^+$-corners
of boxes that represent the active vertices consecutively in the order
that they appear in the outer face of $G_k$.

For step $k\geq 2$,
we add boxes for the vertices in $V_k = \{ v_a, v_{a+1}, \ldots, v_b \}$
with their $z^-$-facets on the plane $z=k-1$ and their $z^+$-facets on
$z=k$.

 If $a=b$, let $w_\ell$ and $w_r$ be the leftmost and rightmost neighbors of $v_a=v_b$ in $C_{k-1}$.
We create the box for $v_a$ so that
its $x^+$-facet touches box $w_\ell$,
its $y^+$-facet touches box $w_r$, and
its $z^-$-facet covers (and touches if $(v_a,w_i) \in E$)) the
 box for $w_i$ ($\ell < i < r$)
(this is possible due to the staircase invariant).
Boxes in this last set are now no longer
active. By the construction of a compatible Schnyder wood, each edge of $v_a$ added
in this step except possibly the edges with $w_\ell$ and $w_r$
are colored $3$ and directed towards $v_a$. The edge $(v_a,w_\ell)$
is colored $1$ and directed towards $w_\ell$, and the edge $(v_a,w_r)$ is
colored $2$ and directed towards $w_r$; note that one or both these edges may
also be colored $3$ and directed towards $v_a$.

If $a < b$, let $w_\ell$ and $w_r$ be the neighbors of $v_a$ and
$v_b$, respectively, on $C_{k-1}$.
We create the boxes for $v_a, v_{a+1}, \dots, v_b$ so that
the $x$-facets or $y$-facets of the boxes for $v_{a+i}$ and
$v_{a+i+1}$ touch.
The $x^+$-facet of box $v_a$ touches box $w_\ell$ and the $y^+$-facet of
box $v_b$ touches box $w_r$.
By the construction of compatible Schnyder wood, for $1\leq i<b-a$,
the edges $(v_{a+i},v_{a+i+1})$ are bi-directional, and the direction from $v_{a+i}$
to $v_{a+i+1}$ is colored $2$, while the other direction is colored $1$. The edges
$(v_a,w_\ell)$ and $(v_b,w_r)$ are directed and colored as in the case where $a=b$.

In both cases:
if $(v_a, w_\ell)$ is a bi-directional edge (with colors $1$ and $3$)
we align the $y^-$-facets of box $v_a$ and box $w_\ell$ (note
$w_\ell$ is no longer active); and
if $(v_b, w_r)$ is a bi-directional edge (with colors $2$ and $3$)
we align the $x^-$-facets of box $v_b$ and $w_r$ (note $w_r$
is no longer active).
We have not yet set the coordinate of the $z^+$-facet of any of the boxes of $V_k$. We simply
extend the boxes of all active vertices in the $+z$ direction, so that the $z^+$-facet is in
the plane $z=k$.

An illustration of the representation obtained by this algorithm is shown in Fig.~\ref{fig:draw}.

Now we describe how to construct the box-contact representation for the dual of $G$.
Consider the orthogonal surface induced by $\Gamma$. Each vertex $v$ of $G$ corresponds
 to the $x^-y^-z^-$-corner $p$ of the box for $v$. The three outgoing edges of $v$ in the Schnyder
 wood can be drawn on the surface from $p$ in the directions $x^+$, $y^+$, $z^+$; see Fig.~\ref{fig:draw}.
Each face of $G$ corresponds to a reflex corner of the orthogonal surface and there is a similar
 (opposite) direction for the outgoing edges in the dual Schnyder wood. The topology of this
 (rigid axial) orthogonal surface is uniquely defined by both the Schnyder wood and its dual~\cite{FZ08}.
Thus, we construct the contact representation $\Gamma'$ for the dual of $G$ using the same algorithm,
 then (after possible scaling) the boundary of $\Gamma'$ exactly matches $B-\Gamma$, where $B$
 is the bounding box of $\Gamma$.
%We can thus find a primal-dual contact representation of $G$ by computing both the primal and the
% dual representations $\Gamma$ and $\Gamma'$, then
We fit $\Gamma$ and $\Gamma'$ together by replacing the three boxes for the three outer vertices of $G$ with a single shell-box, which forms the boundary of the entire representation.

The above construction algorithm takes $\Oh(|V|)$ time. An elementary canonical order of $G$ can be
 computed in linear time~\cite{Kan96} and one can compute a compatible Schnyder wood using the
 definition in linear time as well.
The $x$, $y$ and $z$-coordinates of the boxes are computed in constant time for a vertex of a primal
 graph. Hence, the primal representation of the graph can be found in linear time.
Again from the Schnyder wood of $G$, one can compute a dual Schnyder wood for the dual graph
 in linear time~\cite{FZ08} and by Lemma~\ref{lem:schny-can} a compatible elementary canonical order
 for the dual can also be computed in linear time and the same algorithm can be used to construct the
 dual representation in linear time.
 Finally the primal and the dual representation can be combined together by reflecting

The construction for a primal-dual box-contact representation
 of a $3$-connected plane graph $G$ in Theorem~\ref{thm:box} has another interpretation. Start with the orthogonal surface $S$ corresponding to both a Schnyder wood of $G$ and its dual. This orthogonal surface $S$ creates two half-spaces on either side of $S$: extend boxes from $S$ in these half-spaces, so that
(i)~the boxes are interior-disjoint and (ii)~they induce box-contact representations for the
 primal and dual graphs. However, how to extend the corners and why such a construction yields a proper box-contact representation seems to require the same kinds of arguments that we provided in this section.
A ``proof from the book'' for Theorem~\ref{thm:box}, using topological properties of the orthogonal surface, is a nice open problem.

%%%%%%%%%%%%%%%%%%%%%%%%%%%%%%%%%%%%
%%% - structure of 1-planar graphs
%%% - corollary for prime 1-planar graphs
%%% - outline of procedure for general 1-planar graphs
%%% - structure of near-optimal 1-planar graphs
%%% - proof of general result by induction
%%%%%%%%%%%%%%%%%%%%%%%%%%%%%%%%%%%%

\section{Representations for Optimal 1-Planar Graphs}

In this section we consider contact representation for optimal 1-planar graphs. We first show that
 there exists optimal 1-planar graphs with no box-contact representation. We thus consider contact
 representation for optimal 1-planar graphs with the next simplest axis-aligned 3D object: \LLs.
 We provide a quadratic-time algorithm for representing optimal 1-planar graph with \LLs.

\subsection{Optimal 1-Planar Graphs with no Box-Contact Representation}

Here we prove the following lemma.

\begin{lemma}
\label{lem:K5}{\ \\[-1em]}
 \begin{itemize}
  \item $K_5$ has no proper box-contact representation in 3D.
  \item  There exist optimal $1$-planar graphs that have neither box-contact representation nor shelled
 box-contact representation in 3D.
 \end{itemize}
 \end{lemma}

\begin{proof}
Assume for a contradiction that $K_5$ has a contact representation
 $\Gamma$ with five boxes $B_i$, $i\in\{1,2,3,4,5\}$. We first show that there is a point
 $o$ which is common to all five boxes in $\Gamma$, that is, $\cap_{i=1}^{5}B_i$ is not empty.
 Since $B_i$ is an axis-aligned box, one can define $B_i$ as
 $[x_i,X_i]\times[y_i,Y_i]\times[z_i,Z_i]$ for each $i\in\{1,2,3,4,5\}$.
 Now for any pair of boxes $B_i$, $B_j$, $i,j\in\{1,2,3,4,5\}$, $B_i$ and $B_j$ makes a
 contact in $\Gamma$ since $\Gamma$ represents $K_5$; hence there is a point $(x_{ij}, y_{ij},
 z_{ij})$ common to $B_i$, $B_j$. This can happen only if $x_{ij}\in[x_i,X_i]$, $x_{ij}\in[x_j,X_j]$,
 $y_{ij}\in[y_i,Y_i]$, $y_{ij}\in[y_j,Y_j]$ and $z_{ij}\in[z_i,Z_i]$, $z_{ij}\in[z_j,Z_j]$. In particular
 if we look at the projection on the $x$-axis, $x_i\le x_{ij}\le X_i$ and $x_j\le x_{ij}\le X_j$;
 that is, $[x_i,X_i]$ and $[x_j,X_j]$ have a common point $x_{ij}$. Then by Helly's
 theorem~\cite{AH95}, $[x_i,X_i]$, $i\in\{1,2,3,4,5\}$ have a common
 point $x^*$. Similarly all the projections of $B_i$'s on the $y$-axis (resp. on the $z$-axis)
 have a common point $y^*$ (resp. $z^*$). Then for each box $B_i$, $(x^*,y^*,z^*)\in B_i$
 and therefore $(x^*,y^*,z^*)$ is a common point for all five boxes. Call this point $o$.

\begin{figure}[b!]
\centering
\includegraphics[width=0.3\textwidth]{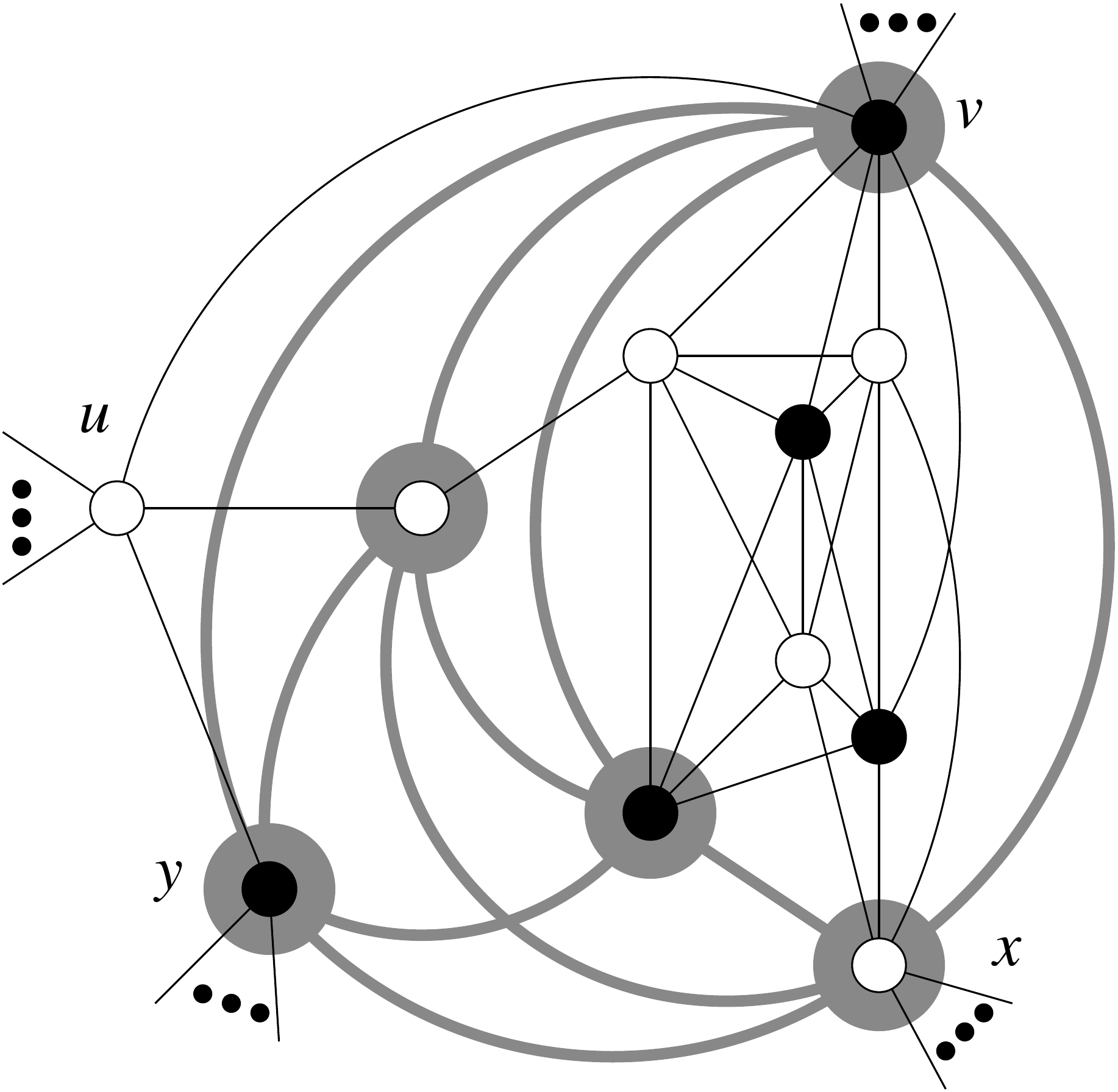}
\caption{An almost optimal $1$-planar graph $G^*$ containing a $K_5$ (highlighted).}
\label{fig:K5}
\end{figure}

Now consider the eight closed octants with $o$ in the center.
Denote the positive and negative $x$-half-space of $o$ by $x_{+}$ and $x_{-}$,
 respectively. Similarly define $y_{+}$, $y_{-}$, $z_{+}$ and $z_{-}$ to be the positive
 $y$-half-space, negative $y$-half-space, positive  $z$-half-space and negative $z$-half-space,
 respectively. Then for $X\in\{x_{+}, x_{-}\}$, $Y\in\{y_{+}, y_{-}\}$, $Z\in\{z_{+}, z_{-}\}$,
 denote the octant $X\cap Y\cap Z$ by $XYZ$. Since all five boxes contain $o$, no octant
 can contain parts of two or more boxes. Since there are eight octants and five boxes,
 by pigeonhole principle, there are at least two boxes, each of which are completely contained
 inside one octant. Assume, \WLOG (after possible renaming), that $B_1$ and
 $B_2$ are these two boxes and they lie on the two octants $x_{+}y_{+}z_{+}$ and
 $x_{+}y_{+}z_{-}$, respectively. Then for any box to make a proper contact with both $B_1$
 and $B_2$, it must lie either in both the octants $x_{-}y_{+}z_{+}$, $x_{-}y_{+}z_{-}$, or in
 both the octants $x_{+}y_{-}z_{+}$, $x_{+}y_{-}z_{-}$. This implies at least one of the three
 boxes $B_3$, $B_4$ and $B_5$ fails to make a proper contact with at least one of $B_1$
 and $B_2$. This completes the proof for $K_5$.

In order to prove the second part of the claim, first note that there exists an almost-optimal $1$-planar
graph containing $K_5$ as a subgraph; see Fig.~\ref{fig:K5}. Call this graph $G^*$. Now consider an optimal $1$-planar
graph $G$. Take two faces $f_1=a_1b_1c_1d_1$ and $f_2=a_2b_2c_2d_2$ in the
quadrangulation of $G$ such that they do not share vertices. Delete the crossing
edges inside the two faces and in each face $f_i$, $i=1,2$, and insert an isomorphic copy
of $G^*$, identifying $u,v,x,y$ with $a_i,b_i,c_i,d_i$. Let $H$ be the resulting augmented graph.
As shown earlier, $H$ has no contact representation with interior-disjoint boxes.
Furthermore, $H$ does not admit a shelled box-contact representation, as
at least one copy of $G^*$ in $H$ does not contain the vertex represented by the shell.
\end{proof}

\subsection{L-Contact Representation of Optimal 1-Planar Graphs}

Here we prove Theorems~\ref{th:prime}~and~\ref{thm:1-planar}.
Throughout, let $G$ be an optimal 1-planar graph with a fixed 1-planar embedding.
An edge is \df{crossing} if it crosses another edge, and \df{non-crossing} otherwise.
A cycle in a connected graph is \df{separating} if removing it disconnects the graph.
We start by listing some well-known properties of optimal 1-planar graphs.
% and refer to Figure~\ref{fig:opt-1-planar} for an illustrating example.

\begin{lemma}[Brinkmann \df{et al.}~\cite{BGGMT+05}, Suzuki~\cite{Suz10}]\label{lem:1-planar-subgraphs}{\ \\[-1em]}
 \begin{itemize}
  \item The subgraph of an embedded optimal 1-planar graph $G$ induced by the non-crossing edges is a plane quadrangulation $Q$ with bipartition classes $W$ and $B$.
  \item The induced subgraphs $G_W = G[W]$ and $G_B = G[B]$ on white and black vertices, respectively, are planar and dual to each other.
  \item The graphs $G_B$ and $G_W$ are $3$-connected if and only if $Q$ has no separating $4$-cycles.
  \item There exists a simple optimal 1-planar graph with quadrangulation $Q$ if and only if $Q$ is $3$-connected.
 \end{itemize}
\end{lemma}

\begin{figure}[htb]
%\vspace{-0.2cm}
  \centering
  \begin{subfigure}[t]{.3\textwidth}
    \centering
    \includegraphics[height=3cm]{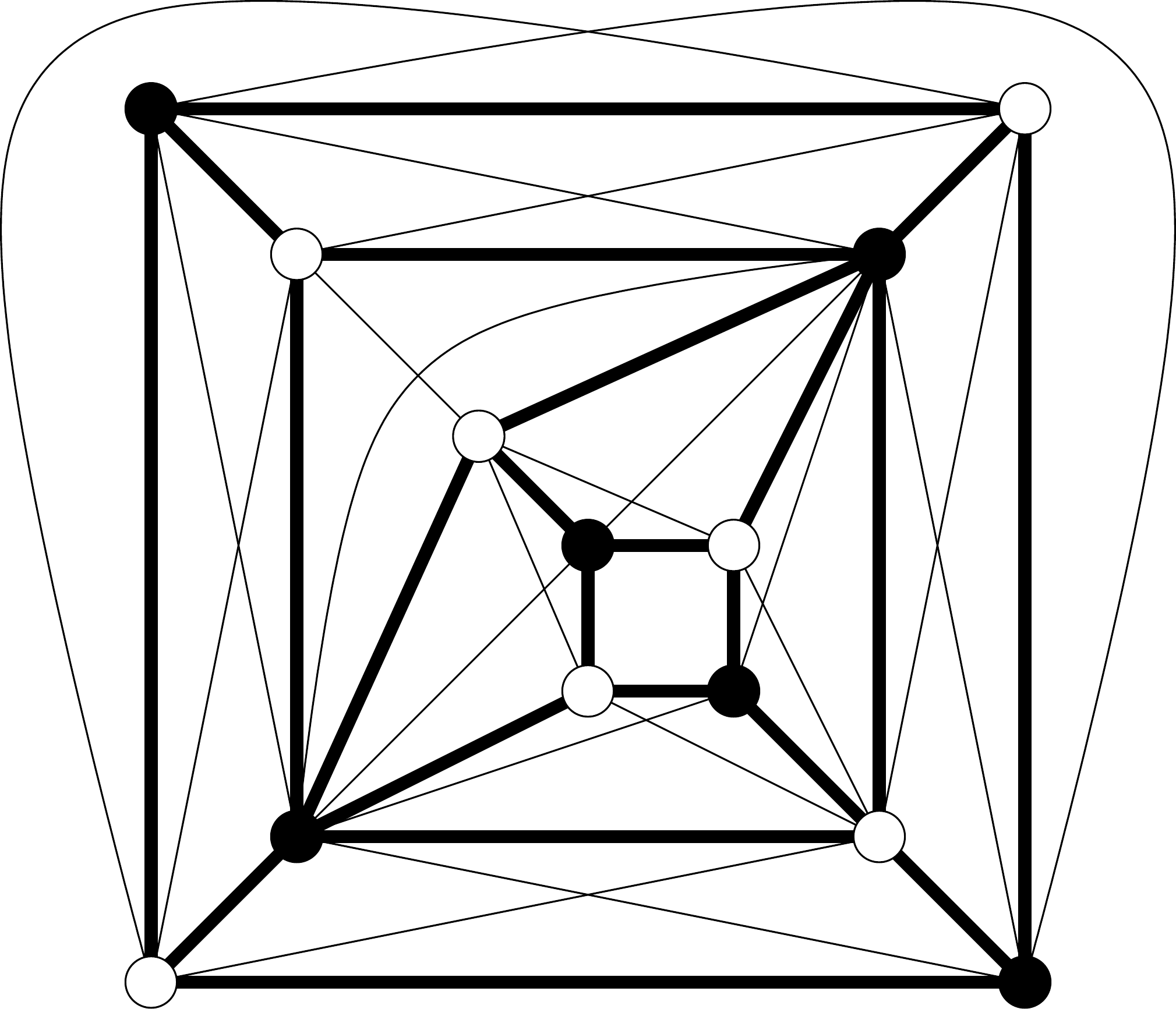}
    \caption{}
    \label{fig:opt-1-planar}
  \end{subfigure}
  \hspace{1em}
  \begin{subfigure}[t]{.3\textwidth}
    \centering
    \includegraphics[height=3cm]{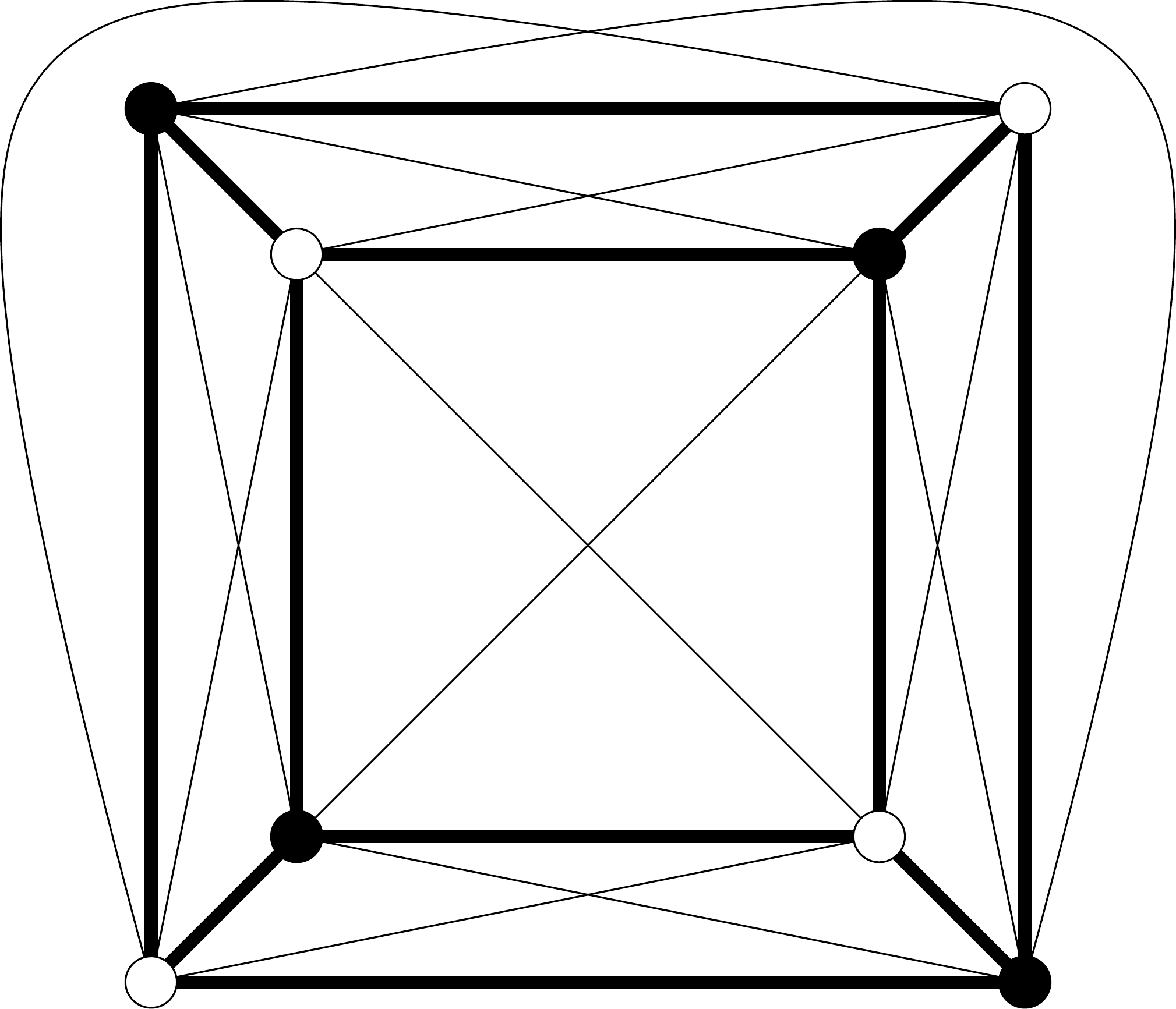}
    \caption{}
    \label{fig:opt-H1}
  \end{subfigure}
  \hspace{2em}
  \begin{subfigure}[t]{.15\textwidth}
    \includegraphics[height=2cm]{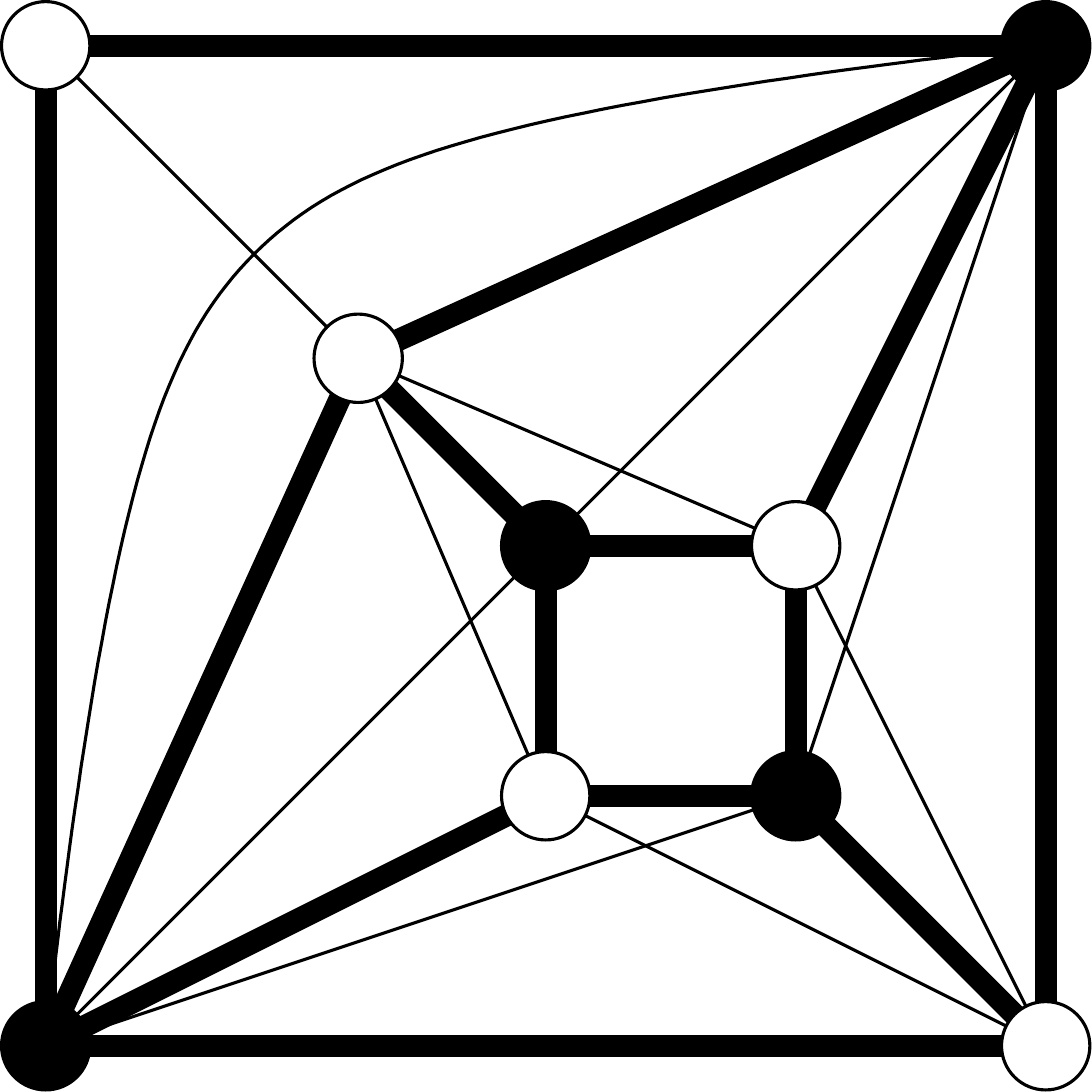}
    \caption{}
    \label{fig:opt-H2}
  \end{subfigure}
  \caption{(\subref{fig:opt-1-planar})~An embedded optimal 1-planar graph, its quadrangulation $Q$ (bold) and the partition into white and black vertices.
  (\subref{fig:opt-H1})~The graph $G_{out}$ produced by removing the interior of separating $4$-cycle $C$.
  (\subref{fig:opt-H2})~The graph $G_{in}(C)$ comprised of the separating $4$-cycle and its interior.}
\end{figure}

%Let us call an optimal 1-planar graph \df{prime} if its corresponding quadrangulation has no separating $4$-cycles.

Call an optimal 1-planar graph \df{prime} if its quadrangulation has no separating $4$-cycle.

\begin{backInTime}{th-prime}
\begin{corollary}\label{cor:goodCase}
Every prime $1$-planar graph $G=(V,E)$ admits a shelled box-contact representation
in 3D and it can be computed in $\Oh(|V|)$ time.
\end{corollary}
\end{backInTime}
\begin{proof}
% Let $G$ be a prime 1-planar graph with corresponding quadrangulation $Q$ and bipartition classes $B$ and $W$.
Let $Q$ be the quadrangulation of $G$ and let $B$, $W$ be the bipartition classes of $Q$.
 By Lemma~\ref{lem:1-planar-subgraphs}, $G_B = G[B]$ and $G_W = G[W]$ are $3$-connected planar and dual to each other.
 By Theorem~\ref{thm:box}, a primal-dual box-contact representation $\Gamma$ of $G_B$
 can be computed in linear time.
 We claim that $\Gamma$ is a contact representation of $G$.
 Indeed, the edges of $G$ are partitioned into $G_B$, $G_W$, $Q$.
 Each edge in $G_B$ is realized by contact of two ``primal'' boxes, in $G_W$ by contact of ``dual'' boxes, and in $Q$ by contact of a primal and a dual box; see Fig.~\ref{fig:big_example}.
\end{proof}

\begin{figure}[htb]
 \centering
 \includegraphics{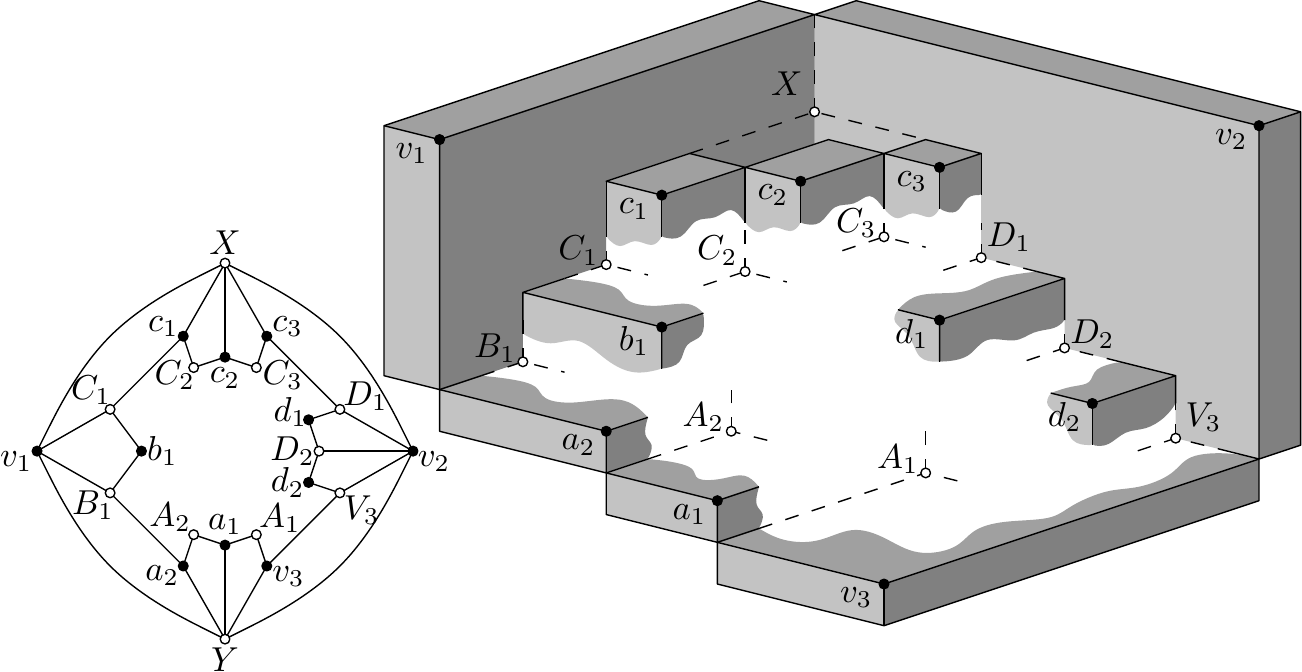}
 \caption{Part of an optimal 1-planar graph and its partial proper box contact representation}
 \label{fig:big_example}
\end{figure}

Next, assume that $G$ is any (not necessarily prime) optimal 1-planar graph.
To find an \LL-representation for $G$, we find all separating $4$-cycles in $G$, replace their interiors by a pair of crossing edges and construct an \LL-representation of the obtained smaller 1-planar graph by Corollary~\ref{cor:goodCase}.
We ensure that this \LL-representation has some ``available space'' in which we can place the \LL-representations for the removed subgraph in each separating $4$-cycle, which we construct recursively.
We remark that similar procedures were used before, e.g., for maximal planar graphs and their separating triangles~\cite{FF11,Tho88}.
A separating $4$-cycle is \df{maximal} if its interior is inclusion-wise maximal among all separating $4$-cycles.
A 1-planar graph with at least $5$ vertices is called \df{almost-optimal} if its non-crossing edges induce a quadrangulation $Q$ and inside each face of $Q$, other than the outer face, there is a pair of crossing edges.

%\noindent
%\textbf{Algorithm for Drawing Optimal $1$-Planar Graph}

\newcommand{\lt}{\textbf{Let} }
\newcommand{\drawOpt}{algorithm \textbf{L-Contact}}

%\vspace{-0.3cm}
\hspace{-0.02\textwidth}
\parbox{0.97\textwidth}
{
\begin{algorithm}{L-Contact}[\text{optimal 1-planar graph }G]
	{
%	\qinput{$G$: an embedded optimal $1$-planar graph}
%	\qoutput{An proper \LL-contact representation of $G$}
	}

%	\lt $Q$ be the quadrangulation of $G$\\

	Find all separating $4$-cycles in the quadrangulation $Q$ of $G$\\

	\qif some inner vertex $w$ of $Q$ is adjacent to two outer vertices of $Q$\\

		\qthen
		$\mathcal{C}=$ set of the two $4$-cycles containing $w$ and 3 outer vertices of $Q$.
		 \textbf{(Case 1)}

		\qelse
		$\mathcal{C}=$ set of all maximal separating $4$-cycles in $Q$.
		\textbf{(Case 2)}
	\qfi\\

	Take the optimal 1-planar (multi)graph $G_{out}$ obtained from $G$ by replacing
	for each $4$-cycle $C \in \mathcal{C}$ all vertices strictly inside $C$
	by a pair of crossing edges; see Fig.~\ref{fig:opt-H1}.
	\label{step:define_G_out}\\

	Compute an \LL-representation of $G_{out}$ that has ``some space'' at each $4$-cycle
	$C \in \mathcal{C}$. In Case~2, this is done from a shelled box-contact representation
	of $G_{out}$ in Corollary~\ref{cor:goodCase}.
	\label{step:G_out}\\

%\vspace{-0.4cm}
%	Go through all $4$-cycles $C \in \mathcal{C}$ and consider the almost-optimal 1-planar
	For each $C \in \mathcal{C}$, take the almost-optimal 1-planar
	subgraph $G_{in}(C)$ of $G$ induced by $C$ and all vertices inside $C$;
	see Fig.~\ref{fig:opt-H2}. Recursively compute an \LL-representation of $G_{in}(C)$
	and insert it into the corresponding ``space'' in the \LL-representation of $G_{out}$.
	\label{step:G_in}

\end{algorithm}
}

\begin{comment}
\begin{figure}[h!]
 \begin{enumerate}
  \item Find all separating $4$-cycles in the quadrangulation $Q$ of $G$.
  \item If \textbf{(Case 1)} some inner vertex $w$ of $Q$ is adjacent to two outer vertices of $Q$ let $\mathcal{C}$ be the set of the two $4$-cycles in $Q$ consisting of $w$ and three outer vertices of $Q$.
  \item Otherwise \textbf{(Case 2)} let $\mathcal{C}$ be the set of all maximal separating $4$-cycles in $Q$.
  \item Consider the optimal 1-planar (multi)graph $G_{out}$ obtained from $G$ by replacing for each $4$-cycle $C \in \mathcal{C}$ all vertices strictly inside $C$ by a single pair of crossing edges; see Fig.~\ref{fig:opt-H1}.\label{step:define_G_out}
  \item Compute an \LL-representation of $G_{out}$ that has ``some space'' at each $4$-cycle $C \in \mathcal{C}$. In Case~2 this is done based on the shelled box-contact representation of $G_{out}$ in Corollary~\ref{cor:goodCase}.\label{step:G_out}
  \item Go through all $4$-cycles $C \in \mathcal{C}$ and consider the almost-optimal 1-planar subgraph $G_{in}(C)$ of $G$ induced by $C$ and all vertices strictly inside $C$; see Fig.~\ref{fig:opt-H2}. Compute an \LL-representation of $G_{in}(C)$ recursively and insert it into the corresponding ``space'' in the \LL-representation of $G_{out}$.\label{step:G_in}
 \end{enumerate}
 \caption{Outline of the procedure constructing an \LL-representation for any given optimal 1-planar graph $G$ in linear time.}
 \label{fig:outline}
\end{figure}
\end{comment}

Let us formalize the idea of ``available space'' mentioned in steps~\ref{step:G_out} and~\ref{step:G_in} in the above procedure.
Let $\Gamma$ be any \LL-representation of some graph $G$ and $C$ be a $4$-cycle in $G$.
A \df{frame for $C$} is a $3$-dimensional axis-aligned box $F$ together with an injective mapping
of $V(C)$ onto the facets of $F$ such that the two facets without a preimage are adjacent.
Every frame has one of two possible types.
If two opposite vertices of $C$ are mapped onto two opposite facets of $F$, then $F$ has type $\ppar$;
otherwise, $F$ has type $\pperp$;  see Fig.~\ref{fig:frames}.
Finally, for an almost-optimal 1-planar graph $G$ with corresponding quadrangulation $Q$ and outer face $C$, and a given frame $F$ for $C$, we say that an \LL-representation $\Gamma$ of $G$ \df{fits into $F$} if replacing the boxes or \LL's for the vertices in $C$ by the corresponding facets of $F$ yields a proper contact representation of $G - E(G[C])$ that is strictly contained in $F$.

Before we prove this Section's main result, namely Theorem~\ref{thm:1-planar}, we need one last lemma addressing the structure of maximal separating $4$-cycles in almost-optimal 1-planar graphs.
% Recall that by Lemma~\ref{lem:1-planar-subgraphs} the quadrangulation of an optimal 1-planar graphs is always $3$-connected.
% For almost-optimal 1-planar graphs, we have the following.

\begin{lemma}\label{lem:almost}
 Let $G$ be an almost-optimal 1-planar graph with corresponding quadrangulation $Q$.
 Then all maximal separating $4$-cycles of $Q$ are interior-disjoint, unless two inner vertices $w$ and $w'$ of $Q$ are adjacent to two outer vertices of $Q$.
\end{lemma}
\begin{proof}
 When two maximal separating $4$-cycles $C$ and $C'$ are not interior-disjoint, then some vertex from $C$ lies strictly inside $C'$ and some vertex from $C'$ lies strictly inside $C$.
 It follows that $V(C) \cap V(C')$ is a pair $x,y$ of two vertices from the same bipartition class of $Q$, say $x,y \in B$, and that some $v \in V(C)$ lies strictly outside $C'$ and some $v' \in V(C')$ lies strictly outside $C$.
 We have $v,v' \in W$ and that $C^* = (x,v,y,v')$ is a $4$-cycle whose interior strictly contains $C$ and $C'$.
 By the maximality of $C$ and $C'$ it follows that $C^*$ is not separating.
 As the vertices $w \in V(C) \setminus V(C^*)$ and $w' \in V(C') \setminus V(C^*)$ lie strictly inside $C^*$, $C^*$ must be the outer cycle of $Q$ and $w,w'$ are the desired vertices.
\end{proof}

\begin{figure}[t]
%\vspace{-0.3cm}
 \centering
  \begin{subfigure}[t]{.2\textwidth}
    \centering
    \includegraphics{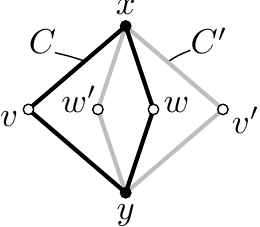}
    \caption{}
    \label{fig:sep-pair-b}
  \end{subfigure}
  \hspace{3em}
  \begin{subfigure}[t]{.5\textwidth}
    \centering
    \includegraphics{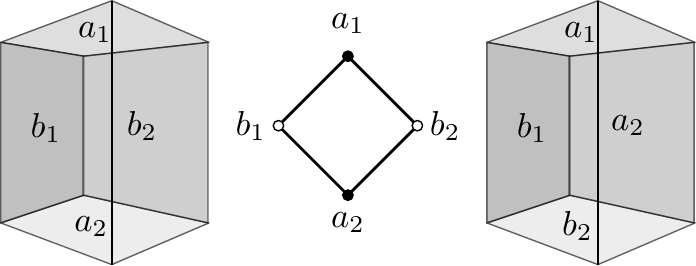}
    \caption{}
    \label{fig:frames}
  \end{subfigure}
  \caption{(a)~Illustration for the proof of Lemma~\ref{lem:almost}.
  (b)~A frame of type $\ppar$ (left) and of type $\pperp$ (right).}
\end{figure}

%\vspace{-0.3cm}
\begin{backInTime}{thm-1-planar}
 \begin{theorem}
  Every optimal 1-planar graph $G=(V,E)$ has an \LL-representation and it can be computed in $\Oh(|V|^2)$ time.
 \end{theorem}
\end{backInTime}
\begin{proof}
 Fix any 1-planar embedding of $G$ and let $Q$ be the corresponding quadrangulation with outer cycle $C_{out}$. Following {\drawOpt},
% the outline in {\drawOpt},
% Figure~\ref{fig:outline}
 we distinguish two cases.
 If (\textbf{Case~1}) some inner vertex $w$ of $Q$ has two neighbors on $C_{out}$ we let $\mathcal{C}$ be the set of the two $4$-cycles in $Q$ that consist of $w$ and 3 vertices of $C_{out}$.
 Otherwise (\textbf{Case~2}), let $\mathcal{C}$ be the set of all maximal separating $4$-cycles in $Q$.
% Note that b
By Lemma~\ref{lem:almost} the cycles in $\mathcal{C}$ are interior-disjoint.
 As in step~\ref{step:define_G_out} we define $G_{out}$ to be the optimal 1-planar (multi)graph obtained from $G$ by replacing for each $C \in \mathcal{C}$ all vertices strictly inside $C$ by a pair of crossing edges.
 Note that in Case~1 the quadrangulation corresponding to $G_{out}$ is $K_{2,3}$ with inner vertex $w$.
 We proceed by proving the following claim, which corresponds to step~\ref{step:G_out} in the algorithm.
% {\drawOpt}.
% Fig.~\ref{fig:outline}.

 \begin{claim}
  Let $H$ be an almost-optimal 1-planar (multi)graph whose corresponding quadrangulation $Q_H$ is either $K_{2,3}$ or has no separating $4$-cycles.
  Let $\mathcal{C}$ be a set of facial $4$-cycles of $Q_H$, different from $C_{o}$, and $H'$ be the graph obtained from $H$ by removing the crossing edges in each $C \in \mathcal{C}$.
  Then for any given frame $F$ for the outer cycle $C_{o}$ of $Q_H$ one can compute an
 \LL-representation $\Gamma$ of $H'$ fitting into $F$ such that for every $C \in \mathcal{C}$ there is a frame $F_C \subseteq F$ for $C$ that is interior-disjoint from all boxes and \LL's in $\Gamma$.
\label{cl:two-case}
 \end{claim}
 \begin{claimproof}{\ref{cl:two-case}}
  \begin{description}
   \item[Case 1, $Q_H = K_{2,3}$.] Let $w$ be the inner vertex of $H$.
    Without loss of generality let $F = [0,5]\times[0,5]\times[0,4]$ and let $V(C_{o})$ be mapped onto the top, back left, bottom and back right facets of $F$.
    We define the \LL for $w$ to be the union of $[0,3]\times[2,3]\times[0,4]$ and $[2,3]\times[0,3]\times[0,4]$.
    Further define four boxes $F_1 = [0,2]\times[0,1]\times[0,1]$, $F_2 = [0,2]\times[0,1]\times[3,4]$, $F_3 = [3,4]\times[0,1]\times[0,4]$ and $F_4 = [0,1]\times[3,4]\times[0,4]$, each completely contained in $F$ and disjoint from the \LL for $w$; see Fig.~\ref{fig:K23}.
Each $F_i$ is a frame for a $4$-tuple containing $w$ and exactly three vertices of $C_{o}$, $i=1,2,3,4$.
    Thus independent of the type of $F$ and the neighbors of $w$ in $Q_H$, we find a frame for both inner faces of $Q_H$.

    \begin{figure}[htb]
     \centering
     \includegraphics[width=0.35\textwidth]{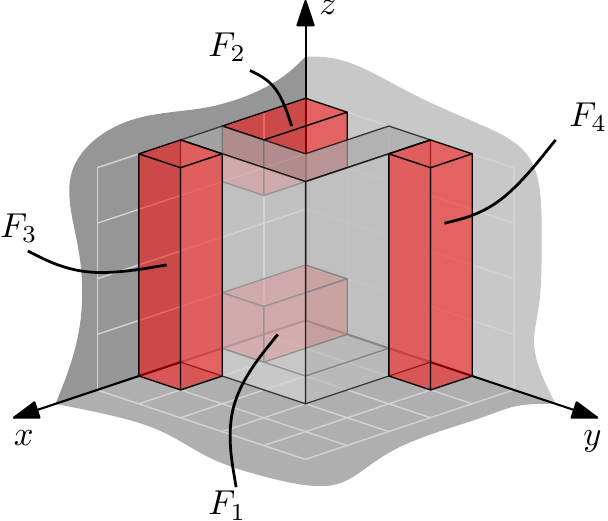}
     \caption{Construction for Case 1 in the proof of Claim~\ref{cl:two-case}.}
     \label{fig:K23}
    \end{figure}

    \item[Case 2, $Q_H \neq K_{2,3}$.]
     Let $B$ and $W$ be the bipartition classes of $Q_H$ and $C_{o} = (v_1,w_1,v_2,w_2)$ with $v_i \in B$ and $w_i \in W$, $i=1,2$.
     Without loss of generality $v_1,v_2,w_1$ are mapped onto the back left, back right and top facets of $F$, respectively, and $w_2$ is mapped onto the bottom facet if (\textbf{Case~2.1}) $F$ has type $\ppar$ and onto the front left facet if (\textbf{Case~2.2}) $F$ has type $\pperp$.
     Let $H^*$ be the graph obtained from $H$ by inserting a pair of crossing edges in $C_{o}$, leaving $v_1,w_2$ and $v_2$ on the unbounded region.
     By assumption, $H^*$ is a prime 1-planar graph and thus by Lemma~\ref{lem:1-planar-subgraphs} $H^*_B = H^*[B]$ and $H^*_W = H^*[W]$ are planar $3$-connected and dual to each other.
     We choose $v_3$ to be the clockwise next vertex after $v_2$ on the outer face of $H^*_B$ and compute (using Corollary~\ref{cor:goodCase}) a shelled box-contact representation $\Gamma^*$ of $H^*$, in which $w_2$ is represented as the bounding box $F^* = [0,n]^3$, $n \in \mathbb{N}$, and $v_1,v_2,w_1$ as $[0,n]\times[0,1]\times[0,n]$, $[0,1]\times[0,n]\times[1,n]\times[1,n]\times[n-1,n]$, i.e., these boxes constitute the back left, back right and top facets of $F^*$, respectively.

     Next we show how to create a frame for each facial $4$-cycle $C \in \mathcal{C}$.
     Let $a_1,b_1,a_2,b_2$ be the vertices of $C$ in cyclic order.
     Assume without loss of generality that $a_1,a_2 \in W$ and $b_1,b_2 \in B$.
     Thus $(a_1,a_2)$ and $(b_1,b_2)$ are crossing edges of $H^*_W$ and $H^*_B$, respectively.
     In the Schnyder wood of $H^*_W$ underlying Corollary~\ref{cor:goodCase} exactly one of $(a_1,a_2)$, $(b_1,b_2)$ is uni-directed, say $(a_1,a_2)$ is uni-directed in tree $\TT_1$.
     Then there is a point in $\mathbb{R}^3$ in common with all four boxes in $\Gamma^*$ corresponding to vertices of $C$.
     Moreover, by Lemma~\ref{lem:box-schny} boxes $b_1, a_2, b_2$ touch box $a_1$ with their $y^+,z^+,y^-$ facets, respectively; see Fig.~\ref{fig:frame_from_good_case}.
     Now we can increase the lower $z$-coordinate of the box $a_1$ by some $\varepsilon > 0$ so that $a_1$ and $a_2$ lose contact and between these two boxes a cubic frame $F_C$ with side length $\varepsilon$ is created; see again Fig.~\ref{fig:frame_from_good_case}.
     Note that by Lemma~\ref{lem:box-schny} the $z^-$ facet of $a_1$ makes contact only with $a_2$ and hence if $\varepsilon$ is small enough all other contacts in $\Gamma^*$ are maintained.
     We apply this operation to each $C \in \mathcal{C}$ and obtain a shelled box-representation $\Gamma'$ of $H'$.

     \begin{figure}[htb]
%\vspace{-0.3cm}
      \centering
      \includegraphics{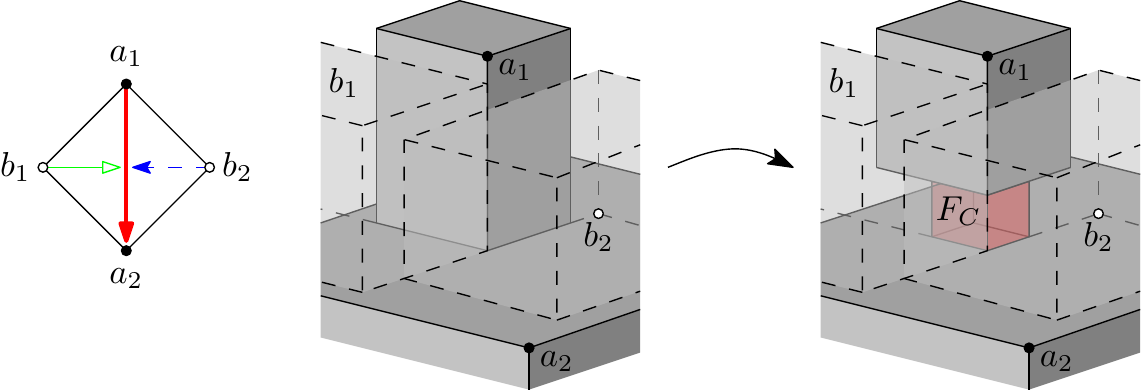}
      \caption{Creating a frame $F_C$ for an inner facial cycle $C = (a_1,b_1,a_2,b_2)$ of $Q_H$ by releasing the contact between $a_1$ and $a_2$.}
      \label{fig:frame_from_good_case}
     \end{figure}

     Finally, we show how to modify $\Gamma'$ to obtain an \LL-representation of $H'$ fitting the given frame $F$.
     If (\textbf{Case~2.1}) $F$ has type $\ppar$, we define a new box for $w_2$ to be $[0,n+1]\times[0,n]\times[-1,0]$.
     For each white neighbor of $w_2$ we union the corresponding box with another box that is contained in $[n,n+1]\times[0,n]\times[0,n]$ with bottom facet at $z=0$ so that the result is an \LL-shape.
     For each black neighbor of $w_2$ we set the lower $z$-coordinate of the corresponding box to $0$; see Fig.~\ref{fig:good_case_into_frame}.
     (This requires the proper contacts for outer edges of $G_B$, except for $(v_1,v_2)$, to be parallel to the $xz$-plane, which we can easily guarantee.)
%     Secondly, we
We then apply an affine transformation mapping $[1,n+1]\times[1,n]\times[0,n-1]$ onto $F$.
     If (\textbf{Case~2.2}) $F$ has type $\pperp$, we define a new box for $w_2$ to be $[0,n]\times[n,n+1]\times[0,n]$ and apply an affine transformation mapping $[1,n]\times[1,n]\times[0,n-1]$ onto $F$.
 %    Clearly, i
In both cases we have an \LL-representation of $H'$ fitting $F$.
  \end{description}\vspace{-0.7cm}
 \end{claimproof}

     \begin{figure}[htb]
      \centering
      \includegraphics{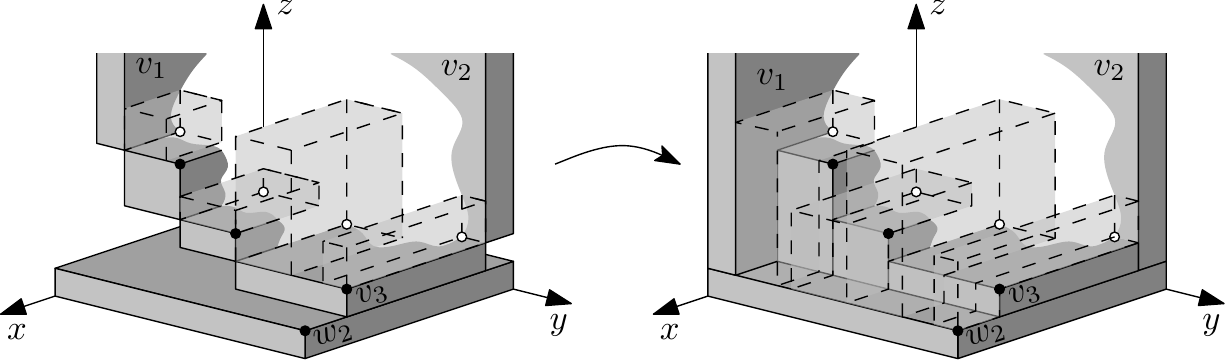}
\caption{Modifying $\Gamma'$ when $F$ has type $\ppar$ (Case~2.1) to
% obtain an \LL-representation fitting $F$.}
 find a representation fitting $F$.}
      \label{fig:good_case_into_frame}
     \end{figure}

By the claim above we can compute an \LL-representation $\Gamma_{out}$ of $G_{out}$ fitting any given
 frame $F_{out}$ for $C_{out}$ in $\Oh(|V(G_{out})|)$ time.
Moreover, $\Gamma_{out}$ has a set of disjoint frames $\{F_C \mid C \in \mathcal{C}\}$.
Following step~\ref{step:G_in} of {\drawOpt},
% the algorithm in Fig.~\ref{fig:outline},
% we go through
% all cycles in $\mathcal{C}$.
for each $C \in \mathcal{C}$, let $G_{in}(C)$ be the almost-optimal 1-planar graph given by all vertices
 and edges of $G$ on and strictly inside $C$.
Recursively applying the claim we can compute an \LL-representation $\Gamma_C$ of $G_{in}(C)$
 fitting the frame $F_C$ for $C$ in $\Gamma_{out}$.
Clearly, $\Gamma = \Gamma_{out} \cup \bigcup_{C \in \mathcal{C}} \Gamma_C$ is an \LL-representation
 of $G$ fitting $F_{out}$.
We pick a frame $F_{out}$ of arbitrary type for $C_{out}$ to complete the construction.
\end{proof}

\begin{comment}

\begin{figure}[t]
  \begin{subfigure}[t]{.4\textwidth}
    \centering
    \includegraphics[height=3.5cm]{cut-box2a}
    \caption{}
    \label{fig:cut-boxa}
  \end{subfigure}
  \hspace{2em}
  \begin{subfigure}[t]{.4\textwidth}
    \centering
    \includegraphics[height=3.5cm]{cut-box2b}
    \caption{}
    \label{fig:cut-boxb}
  \end{subfigure}
    \caption{(a,b) Creating a $\ppar$ open box in a representation of an optimal 1-planar graph.}
\end{figure}

\end{comment}

\section{Conclusion and Open Questions}
In this paper we presented new results about primal-dual contact representations in 3D. In particular, we showed that a 3-connected planar graph and its dual has a box-contact representation and that an optimal 1-planar graph has an \LL-contact representation. Many interesting problems remain open.
%We also took a small step beyond planarity and provided an algorithm for representing
%optimal $1$-planar graphs. There are many interesting open questions.

\begin{enumerate}
%\vspace{-0.3cm}
\item Representing graphs with contacts of constant-complexity 3D shapes, such as \LLs, is open for many graph
classes with a linear number of edges, such as $1$-planar graphs, quasi-planar graphs and other nearly planar
graphs. In particular, does there exist an \LL-contact representation of every $1$-planar graph?

\item In 2D, a planar graph has a contact representation with rectangles if and only if it
contains no separating triangle.
%What is the characterization of graphs with 3D box-contact representations?
Which graphs have 3D box-contact representations?

\item It is known that any planar graph admits a proper contact representation
with boxes in 3D and a non-proper contact representation with cubes (boxes with equal sides).
Does every planar graph admit a proper contact representation with cubes?

\item Given an orthogonal surface $S$ corresponding to the Schnyder wood of a $3$-connected plane graph, how can one extend $S$ into a primal-dual box-contact representation using just topological properties of $S$?
\end{enumerate}

%\medskip\noindent{\bf Acknowledgments:}
\section*{Acknowledgments}
We thank Michael Bekos, Therese Biedl, Franz Brandenburg, Michael Kaufmann, Giuseppe Liotta for useful discussions about different variants of these problems.

%\newpage
%\bibliographystyle{abbrv}
%\bibliography{contact}

\end{document}